\documentclass[pra,twocolumn,showpacs,preprintnumbers,amsmath,amssymb,superscriptaddress]{revtex4}

\usepackage{bm}
\usepackage{graphicx}
\usepackage{amsbsy}
\usepackage{amsmath}
\usepackage{amsfonts}
\usepackage{amsthm}

\begin{document}

\theoremstyle{plain}
\newtheorem{theorem}{Theorem}
\newtheorem{lemma}[theorem]{Lemma}
\newtheorem{corollary}[theorem]{Corollary}
\newtheorem{conjecture}[theorem]{Conjecture}
\newtheorem{proposition}[theorem]{Proposition}

\theoremstyle{definition}
\newtheorem{definition}{Definition}

\theoremstyle{remark}
\newtheorem*{remark}{Remark}
\newtheorem{example}{Example}

\title{All Maximally Entangled Four Qubits States}   

\author{Gilad Gour}\email{gour@math.ucalgary.ca}
\affiliation{Institute for Quantum Information Science and 
Department of Mathematics and Statistics,
University of Calgary, 2500 University Drive NW,
Calgary, Alberta, Canada T2N 1N4} 
\author{Nolan R. Wallach}\email{nwallach@ucsd.edu}
\affiliation{Department of Mathematics, University of California/San Diego, 
        La Jolla, California 92093-0112}

\date{\today}

\begin{abstract}
We find an operational interpretation for the 4-tangle as a type of residual entanglement, somewhat similar to the interpretation of the 3-tangle. Using this remarkable interpretation, we are able to find the class of maximally entangled four-qubits states which is characterized by four real parameters. The states in the class are maximally entangled in the sense that their average bipartite entanglement with respect to all possible bi-partite cuts is maximal.  We show that while all the states in the class maximize the average tangle, there are only few states in the class that maximize the average Tsillas or Renyi $\alpha$-entropy of entanglement. Quite remarkably, we find that up to local unitaries, there exists two unique states, one maximizing the average $\alpha$-Tsallis entropy of entanglement for \emph{all} $\alpha\geq 2$, while the other maximizing it for all $0<\alpha\leq 2$ (including the von-Neumann case of $\alpha=1$). Furthermore, among the maximally entangled four qubits states, there are only 3 maximally entangled states that have the property that for 2, out of the 3 bipartite cuts consisting of 2-qubits verses 2-qubits, the entanglement is 2 ebits and for the remaining bipartite cut the entanglement between the two groups of two qubits is 1ebit. The unique 3 maximally entangled states are the 3 cluster states that are related by a swap operator. We also show that the cluster states are the only states (up to local unitaries) that maximize the average $\alpha$-Renyi entropy of entanglement for \emph{all} $\alpha\geq 2$.
\end{abstract}  


\maketitle

\section{Introduction}

Entanglement lies at the heart of quantum physics.
It was clear immediately after the discovery of quantum mechanics that entanglement is not ``one but rather the characteristic trait of quantum mechanics, the one that enforces its entire departure from classical lines of thought"~\cite{Sch35}. Nevertheless, it was not until recently, that entanglement, besides of being interesting from a fundamental point of view, was also recognized as a valuable resource for two-party communication tasks such as teleportation~\cite{Ben93} and superdense coding~\cite{Ben92}. 
With the emergence of quantum information science in recent years, much effort has been given to the study of bipartite entanglement~\cite{Hor09}; in particular, to its characterization, manipulation and quantification~\cite{Ple07}. It was realized that maximally entangled states are the most desirable resources for many quantum information processing (QIP) tasks.
While two-party entanglement was very well studied, entanglement in multi-party systems is far less understood, and even the identification of maximally entangled states in multi-party systems is a highly non-trivial task. 

The understanding of highly entangled multi-qubit states is crucial for the implementation of many QIP tasks in quantum networks.
Highly entangled multi-qubit states, such as the cluster states or graph states, are the key resource of one-way or measurement based quantum computer~\cite{Bri01}, and as such raised enormous interest in the QIP community. Even experimental realizations of one-way quantum computing with four-qubit cluster states has been demonstrated successfully~\cite{Wal05}.
Highly entangled multi-qubit states are also the key ingredients of various quantum error correction codes and quantum communication protocols~\cite{Cle99,Sch99,Gou07}.
However, unlike bipartite entanglement, very little is known about the characterization of entanglement in multi-qubits systems.

The complexity in the characterization of entanglement in multi-partite systems can already be seen in the fact that for 3 qubits
there are essentially two types of genuine 3-partite entanglement and even the notion of maximally entangled states is 
not unique~\cite{Vid00}.
One can think of both the GHZ-state, $|\text{GHZ}\rangle=(|000\rangle+|111\rangle)/\sqrt{2}$, and the W-state 
$|W\rangle=(|001\rangle+|010\rangle+|100\rangle)/\sqrt{3}$ as two types of maximally entangled states, that are not related by stochastic
local operations and classical communication (SLOCC). Nevertheless, in the case of 3-qubits, one can single out the GHZ-state as the unique
maximally entangled state for the following two reasons. First, the GHZ class (i.e. the set of states that can be obtained from the GHZ state by SLOCC)
is dense in the space of 3-qubits and therefore the W-class is of measure zero.
This means that by LOCC it is possible to convert the GHZ state to a state that is arbitrarily close to the W-state, but the W-state can not be converted (not even by SLOCC)
to a state that is close to the GHZ-state. Second, the GHZ-state is the only 3-qubit state with the property that the bipartite entanglement between any one qubit and the other two qubits is maximal; that is, the reduced density matrix obtained after the tracing out of any two qubits is proportional to the identity. 

Similarly, for n-qubits one can define maximally entangled states as states with the property that the reduced density matrix obtained after the tracing out of any $k$ qubits, with 
$n/2\leq k\leq n-1$, is proportional to the identity. For example, the codeword states of the 5-qubits error correcting codes are maximally entangled~\cite{Ben96}. However, as we show below,
for four qubits such states do not exist. It is also known~\cite{Rains} that maximally entangled states exist for $n=6$ and do not exist for $n\geq 8$. To the authors knowledge, the case of $n=7$ is unknown.

In this paper we find an operational interpretation of the 4-tangle which enable us to 
characterize all maximally entangled four-qubits states.
We define a state to be maximally entangled if its average bipartite entanglement with respect to all possible bi-partite cuts is maximal (e.g. see~\cite{Love06,Sco04} and references therein). More precisely, we divide the four qubits into two groups, each consisting of two qubits, and calculate the pure bipartite entanglement between the two groups of qubits. We then find the class of all states that maximize the average entanglement
of the $3={4 \choose 2}/2$ such (inequivalent) bi-partite cuts. We find that when we take the measure of bipartite entanglement to be the tangle, there is a 4-real parameter class of states $\mathcal{M}$ that maximize the average tangle. 
However, when we take the measure to be the entropy of entanglement, or the Tsallis and Renyi $\alpha-$entropy of entanglement, we get that up to local unitary there are only two states that maximize the average $\alpha-$entropy of entanglement. Quite remarkably, we find that up to local unitaries the state
\begin{equation}
|L\rangle=\frac{1}{\sqrt{3}}\left[u_0+\omega u_1+\omega^2 u_2\right]\;,
\label{statel}
\end{equation}
where $\omega=e^{i2\pi/3}$,
\begin{align}
& u_0\equiv|\phi^{+}\rangle|\phi^{+}\rangle\;\;,\;\;u_1\equiv|\phi^{-}\rangle|\phi^{-}\rangle\nonumber\\
& u_2\equiv|\psi^{+}\rangle|\psi^{+}\rangle\;,\;\;
u_3\equiv|\psi^{-}\rangle|\psi^{-}\rangle\nonumber
\end{align}
and $|\phi^{\pm}\rangle=(|00\rangle\pm|11\rangle)/\sqrt{2}$ and $|\psi^{\pm}\rangle=(|01\rangle\pm|10\rangle)/\sqrt{2}$,
is the only state that maximize the average Tsallis $\alpha-$entropy of entanglement for \emph{all} $\alpha> 2$ (see Fig.~1).
Interesting properties of the state $|L\rangle$ have been discussed in~\cite{Love06,Jens}.

On the other hand, we show that the state
\begin{equation}
|M\rangle  =\frac{i}{\sqrt{2}}u_0+\frac{1}{\sqrt{6}}\left[u_1+u_2+u_3\right]
\label{entropystates}
\end{equation}
is the only state that maximize the Tsallis $\alpha-$entropy of entanglement for all $0<\alpha< 2$.
Ten years ago the state $|M\rangle$  was conjectured to maximize the entropy of entanglement ~\cite{HS}.
More recently, it was proved that locally it is indeed maximally entangled~\cite{BH}. In Fig.~1 we draw a graph of
the average Tsallis $\alpha$-entropy of entanglement as a function of $\alpha$ for the states $|M\rangle$ and $|L\rangle$.

\begin{figure}[tp]
\includegraphics[scale=.48]{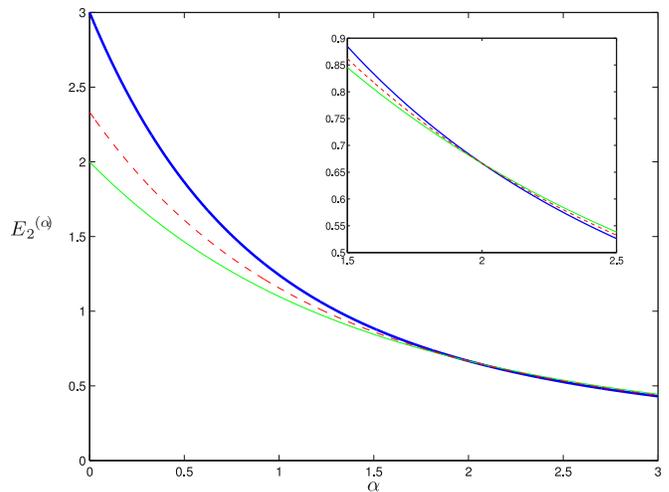}
\caption{A graph of the average Tsallis $\alpha$-entropy of entanglement as a function of $\alpha$. The blue line corresponds to the state $|M\rangle$, the green line to the state $|L\rangle$, and the dashed red line to the cluster states. Like the cluster states, the graph for any maximally entangled state in $\mathcal{M}$ is between the blue and green lines.}
\end{figure}  

In addition, among the maximally entangled four qubits states, we identify 
3 ultimate maximally entangled states that have the property that for 2, out of the 3 bipartite cuts, the entanglement is 2 ebits and for the last bipartite cut the entanglement between the groups of two qubits is 1ebit. The unique 3 maximally entangled states are the 3 cluster states that are related by a swap operator (but not by SLOCC):
\begin{align}
& |C_1\rangle=\frac{1}{2}\left(|0000\rangle+|1100\rangle+|0011\rangle-|1111\rangle\right)\label{C1}\\
& |C_2\rangle=\frac{1}{2}\left(|0000\rangle+|0110\rangle+|1001\rangle-|1111\rangle\right)\label{C2}\\
& |C_3\rangle=\frac{1}{2}\left(|0000\rangle+|1010\rangle+|0101\rangle-|1111\rangle\right)\label{C3}
\end{align}
We show that these cluster states are the only states that maximize the Renyi $\alpha-$entropy of entanglement 
for all $\alpha\geq 2$.

This paper is organized as follows, in section (II) we discuss the generic class of four qubits states, consisting
of an uncountable number of SLOCC-inequivalent classes. In section (III) we find an operational interpretation of the 
4-tangle and discover a four real parameter class
of all four qubits states that maximize the average tangle. We then use this result in section (IV) to find maximally
entangled states with respect to other measures of entanglement, such as the Tsallis and Renyi $\alpha-$entropy
of entanglement. In section (V) we discuss more maximally entangled four qubits states. We end in section (VI)
with a summary, conclusions and a discussion on the extension of the results presented here in higher dimensions.

\section{Uncountable number of four qubits SLOCC-inequivalent classes}

In~\cite{Ver02} it was argued that 4-qubits pure states can be classified into nine groups of states.
One of these nine groups is called the generic class as with the action of SLOCC it is dense in the space of 4-qubits 
$\mathcal{H}_4\equiv \mathbb{C}^2\otimes\mathbb{C}^2\otimes\mathbb{C}^2\otimes\mathbb{C}^2$.
The generic class is given by
$$
\mathcal{A}\equiv\Big\{z_0u_0+z_1u_1+z_2u_2+z_3u_3\Big|\;z_0,z_1,z_2,z_3\in\mathbb{C}\Big\}\;.
$$ 

In~\cite{Ver02,W} it has been shown that all the states that are connected to the class  $\mathcal{A}$ by SLOCC
form a dense set of states. That is, the class of states $G\mathcal{A}$, where 
$G\equiv \text{SL}(2,\mathbb{C})\otimes\text{SL}(2,\mathbb{C})\otimes\text{SL}(2,\mathbb{C})\otimes\text{SL}(2,\mathbb{C})$, is dense in $\mathcal{H}_4$. In the following we discuss several properties of the generic class that will be 
very useful for our theorems in the next sections. 

For $k=0,1,2,3$, we denote by $|k\rangle\rangle\equiv|ij\rangle$, with $i,j=0,1$, a state of two qubits, such that $ij$ is the binary representation of $k$. Hence, any $|\psi\rangle\in\mathcal{H}_4$ can be written as
\begin{equation}
|\psi\rangle=\sum_{k=0}^{3}\sum_{k'=0}^{3}T_{kk'}|k\rangle\rangle|k'\rangle\rangle\;,
\label{T}
\end{equation}
where $\{|k\rangle\rangle\}$ is the computational basis of qubits 1 and 2, and $\{|k'\rangle\rangle\}$ is the computational
basis of qubits 3 and 4. With these notations we define the following four quantities:
\begin{definition}
Let $|\psi\rangle\in\mathcal{H}_4$. Then,
$$
\mathcal{E}_m\left(|\psi\rangle\right)\equiv\left\{ 
                  \begin{array}{rll}
			\text{Det}\left[T_{\psi}\right]\;\;\;\;\;\;\;\;\;\;\;\;\;\;\;\;&\text{if }m=0 \\
			\text{Tr}\left[\left(T_{\psi}JT_{\psi}^{T}J\right)^m\right] &\text{if }m=1,2,3
			
		\end{array}\right.
$$
where $T_{\psi}$ is the $4\times 4$ matrix whose components $T_{kk'}$ are defined in Eq.~(\ref{T}), and
$$
J\equiv\left[
\begin{array}
[c]{cccc}%
0 & 0 & 0 & 1\\
0 & 0 & -1 & 0\\
0 & -1 & 0 & 0\\
1 & 0 & 0 & 0
\end{array}
\right]
$$
\end{definition}
The four polynomials defined above take a simple form on $\mathcal{A}$. If $|\psi\rangle=\sum_{j=0}^{3}z_ju_j$ then
\begin{equation}
\mathcal{E}_m(|\psi\rangle)=\left\{ 
                  \begin{array}{rll}
		           z_0z_1z_2z_3 \;\;\;\;\;\;\;\;\;\;\;\;\;\;\;\;\;\;\;\;\;\;\;\;&\text{if }m=0 \\
			       z_{0}^{2m}+z_{1}^{2m}+z_{2}^{2m}+z_{3}^{2m} &\text{if }m=1,2,3
			
		\end{array}\right.\nonumber
\end{equation}
In~\cite{Ver02,W} it has be shown that these four polinomials are invariant under the action of the group 
$G\equiv\text{SL}(2,\mathbb{C})\otimes\text{SL}(2,\mathbb{C})\otimes\text{SL}(2,\mathbb{C})\otimes\text{SL}(2,\mathbb{C})$.
That is, if $g\in G$ and $|\psi\rangle\in\mathcal{H}_4$ then $\mathcal{E}_m(g|\psi\rangle)=\mathcal{E}_m(|\psi\rangle)$ for
all $m=0,1,2,3$. Other polynomials that corresponds to true `tangles' have been considered for example in~\cite{Jens}.
As discussed in~\cite{Ver02,Ver03}, one of the consequences of this property is that
the four functions $f_m\equiv\big |\mathcal{E}_m\big |^{1/m}$ $(m=1,2,3)$ and $f_0\equiv\sqrt{|\mathcal{E}_{0}|}$ are entanglement monotones~\footnote{Note that $f_0$ is the G-concurrence~\cite{Gour} between qubits (1,2) and (3,4)}. However, here we show that this property implies that almost all the states in $\mathcal{A}$ are not related
by SLOCC, which means that $\mathcal{A}$ contains an uncountable number of SLOCC inequivalent classes of states.
\begin{proposition}(c.f.~\cite{W,Kly02,Ver03})
Let $|\psi\rangle$ and $|\psi'\rangle$ be two normalized states in $\mathcal{A}$. Then, 
the transformation $|\psi\rangle\rightarrow |\psi'\rangle$ can be achieved by SLOCC only if
$f_{m}\left(|\psi\rangle\right)=f_{m}\left(|\psi'\rangle\right)$ for all $m=0,1,2,3$.
\end{proposition}
\begin{proof}
Let $\psi\in\mathcal{A}$ be a normalized state and let $g\in G$. Then,
$$
f_m\left(\frac{g\psi}{\|g\psi\|}\right)=\frac{f_m(g\psi)}{\|g\psi\|^2}=\frac{f_m(\psi)}{\|g\psi\|^2}\leq f_m(\psi)\;,
$$
where in the last inequality we used the fact that for $\psi\in\mathcal{A}$ and $g\in G$,
$\|g\psi\|\geq\|\psi\|=1$ (see Appendix~\ref{KempfNess} for the Kempf-Ness theorem~\cite{KN}).
\end{proof}

In the definition above $\mathcal{E}_m$ is defined only for $m\leq 3$. The absolute value of the polynomials with higher values of $m$ are also entanglement monotones, but they are in the algebra generated by these four polynomials and therefore do not contain any additional information about the entanglement of the states. 

The proposition above implies that the class $\mathcal{A}$ contains an uncountable number of states that are not connected by SLOCC transformation.
More precisely, if $|\psi\rangle=\sum_{j=0}^{3}z_ju_j$ and $|\psi'\rangle=\sum_{j=0}^{3}z_j'u_j$, and there is no permutation $\sigma$ 
such that $z_j= \pm z_{\sigma (j)}'$ for $j=0,1,2,3$ with an even number of $-$ signs, then the transformation 
$|\psi\rangle\rightarrow |\psi'\rangle$ can {\it not} be achieved by SLOCC. Moreover, from the proposition below it follows that on $\mathcal{A}$ if two states are connected
by SLOCC operation, then the transformation must be a local unitary.
\begin{proposition}~\cite{W,Ver03,Kly02}\label{localunitary1}
Let $|\psi\rangle,|\psi'\rangle\in\mathcal{A}$. Then, 
the transformation $|\psi\rangle\rightarrow |\psi'\rangle$ can be achieved by a local unitary 
$U\in\text{SU}(2)\otimes \text{SU}(2)\otimes\text{SU}(2)\otimes\text{SU}(2)$
if and only if 
$\mathcal{E}_{m}\left(|\psi\rangle\right)=\mathcal{E}_{m}\left(|\psi'\rangle\right)$.
\end{proposition}

For the purpose of this work, we generalize the proposition above to include all local unitaries; that is, not only those in
$\text{SU}(2)\otimes \text{SU}(2)\otimes\text{SU}(2)\otimes\text{SU}(2)$. 

\begin{proposition}\label{localunitary}
Set
$$ 
f_4 = |\mathcal{E}_{1}^{2} - \mathcal{E}_0|^2,\;f_5 = |\mathcal{E}_{1}^{2} - \mathcal{E}_2|^2,\;
f_6 = |\mathcal{E}_{1}^{3} - \mathcal{E}_3|^2\;.
$$
Let $|\psi\rangle,|\psi'\rangle\in\mathcal{A}$. Then, 
the transformation $|\psi\rangle\rightarrow |\psi'\rangle$ can be achieved by a local unitary 
$U\in\text{U}(2)\otimes \text{U}(2)\otimes\text{U}(2)\otimes\text{U}(2)$
if and only if 
$f_{m}\left(|\psi\rangle\right)=f_{m}\left(|\psi'\rangle\right)$ for all integers $0\leq m\leq 6$.
\end{proposition}
\begin{proof}
Note that the first 4 conditions (i.e. $m=0,1,2,3$) imply that
\begin{align}
& \mathcal{E}_0(\psi) = a \mathcal{E}_0(\psi')\;,\;
\mathcal{E}_1(\psi) = b \mathcal{E}_1(\psi')\nonumber\\
& \mathcal{E}_2(\psi) = c \mathcal{E}_2(\psi')\;,\;
 \mathcal{E}_3(\psi) = d \mathcal{E}_3(\psi')\;,\nonumber
\end{align}
with $|a|=|b|=|c|=|d|=1$. 
$f_4(\psi)=f_4(\psi')$ implies that $a = b^2$. Similarly, the condition on $f_5$ implies $c=b^2$, and the condition
on $f_6$ implies $d = b^3$. Now, write $b = r^2$. We therefore have
\begin{align}
& \mathcal{E}_0(\psi) = \mathcal{E}_0(r\psi')\;,\;
\mathcal{E}_1(\psi) =  \mathcal{E}_1(r\psi')\nonumber\\
& \mathcal{E}_2(\psi) =  \mathcal{E}_2(r\psi')\;,\;
\mathcal{E}_3(\psi) =  \mathcal{E}_3(r\psi')\;.\nonumber
\end{align}
Thus, from proposition~\ref{localunitary1}, $|\psi\rangle$ and $r|\psi'\rangle$ are related by a local unitary in 
$\text{SU}(2)\otimes \text{SU}(2)\otimes\text{SU}(2)\otimes\text{SU}(2)$. The argument clearly can run backwards.
\end{proof}

As we show now, among the 4 entanglement monotones $f_m$ $(m=0,1,2,3)$, the 4-qubits entanglement monotone $f_1$ 
is the only one that is invariant under any permutation of the 4-qubits. In fact, we find that this monotone is the 4-tangle.

\subsection{The monotone $f_1\equiv |\mathcal{E}_1|$ and the 4-tangle}
Given a bipartite state $|\psi^{AB}\rangle\in\mathbb{C}^{n}\otimes\mathbb{C}^{m}$, the tangle (or the square
of the I-concurrence) is defined by
\begin{equation}\label{deftangle}
\tau_{AB}\equiv\tau\left(|\psi^{AB}\rangle\right)= S_{L}\left(\rho_{r}\right)= 2\left(1-\text{Tr}\rho_{r}^{2}\right)\;,
\end{equation}
where $\rho_{r}=\text{Tr}_{B}|\psi^{AB}\rangle\langle\psi^{AB}|$ is the reduced density matrix, and $S_{L}$ is the linear
entropy. 

For two qubits the tangle can be expressed as the square of the concurrence; that is, for
$|\varphi\rangle\in\mathbb{C}^{2}\otimes\mathbb{C}^{2}$ the tangle is
$$
\tau_{AB}=|\langle\varphi|\tilde{\varphi}\rangle|^2\;\;\text{where}\;\;|\tilde{\varphi}\rangle\equiv 
\sigma_y\otimes\sigma_y|\varphi^{*}\rangle\;,
$$
and $\sigma_y$ is the second Pauli matrix. Note that the basis is chosen such that 
$\sigma_y=\left(\begin{array}[c]{cc}
0 & i\\
-i & 0
\end{array}\right)
$, and if in this basis $|\varphi\rangle=\sum_{i,j}a_{ij}|ij\rangle$ then $|\varphi^{*}\rangle=\sum_{i,j}a_{ij}^{*}|ij\rangle$.

For mixed two qubits state $\rho^{AB}$, the tangle is defined in terms of the convex roof extension:
$$
\tau_{AB}\equiv\tau(\rho^{AB})\equiv\min\sum_{i}p_i\tau(|\psi_i\rangle)\;,
$$ 
where the minimum is taken over all the decompositions of the form $\rho^{AB}=\sum_{i}p_i|\psi_i\rangle\langle\psi_i|$.

In~\cite{CKW} it was shown that one can extend the definition of the 2-qubits tangle to 3-qubits. Given a 3-qubits
pure state $|\psi\rangle\in\mathbb{C}^2\otimes\mathbb{C}^2\otimes\mathbb{C}^2$ the 3-tangle is defined by
$$
\tau_{ABC}\equiv\tau\left(|\psi\rangle\right)\equiv \tau_{A(BC)}-\tau_{AB}-\tau_{AC}\;,
$$
where $\tau_{AB}\equiv\tau(\rho^{AB})$ (with $\rho^{AB}\equiv\text{Tr}_C|\psi\rangle\langle\psi|$), and
$\tau_{A(BC)}$ is the tangle between the qubit system A and two qubits system BC. In~\cite{CKW} it was
shown that the 3-tangle is non-negative and its square root has been proved to be an entanglement monotone
in~\cite{Ver03}. It was also shown~\cite{CKW} that it is symmetric
under permutations of the three qubits A, B, and C. From its definition, the 3-tangle can be interpreted as 
the residual entanglement between A and BC, that can not be accounted for by the entanglements of A and B,
and A and C, separately.

The Wong-Christensen 4-tangle~\cite{Uh00} is defined similarly. Let 
$|\psi\rangle\in\mathcal{H}_4\equiv\mathbb{C}^2\otimes\mathbb{C}^2\otimes\mathbb{C}^2\otimes\mathbb{C}^2$, the 4-tangle is 
defined by~\cite{Uh00}
$$
\tau_{ABCD}\equiv |\langle\psi|
\sigma_y\otimes\sigma_y\otimes\sigma_y\otimes\sigma_y|\psi^{*}\rangle|^2\;.
$$
In~\cite{Uh00} the 4-tangle was shown to be an entanglement monotone and invariant under permutations.
In the next section we will see that like the 3-tangle, the above 4-tangle can also be interpreted as a type
of residual entanglement. Moreover, as we show now, the square of the monotone $f_1$ is the Wong-Christensen 4-tangle.
\begin{proposition}\label{4tangle}
Let $|\psi\rangle\in\mathcal{H}_4$. Then,
$$
\tau_{ABCD}(|\psi\rangle)=|\mathcal{E}_1(|\psi\rangle)|^2\;.
$$
\end{proposition} 
The above proposition follows directly from the fact that there is a single $SL(2,\mathbb{C})^{\otimes 4}$ invariant
polynomial with homogeneous degree 2~\cite{Luq03}. Since both the 4-tangle and $|\mathcal{E}_1|^2$ have these properties
they must be equal. In the proof below we show this equivalence by a direct calculation.
\begin{proof}
Denote $|\psi\rangle=\sqrt{p_0}|0\rangle|\phi^0\rangle+\sqrt{p_1}|1\rangle|\phi^1\rangle$, where 
$|\phi^j\rangle\in\mathbb{C}^2\otimes\mathbb{C}^2\otimes\mathbb{C}^2$ with $j=0,1$ are three qubits orthonormal states.
Also, denote 
$$
|\phi^j\rangle=\sum_{i\in\{0,1\}^{3}}a^{(j)}_i|i\rangle
$$ 
for $j=0,1$.
With this notations it is straightforward to show that both $\mathcal{E}_1(|\psi\rangle)$ and $\langle\tilde{\psi}|\psi\rangle$
equals 
\begin{align}
& 2\sqrt{p_0p_1}  \Big(a_{000}^{(0)}a_{111}^{(1)}+a_{110}^{(0)}a_{001}^{(1)}
	+a_{101}^{(0)}a_{010}^{(1)}+a_{011}^{(0)}a_{100}^{(1)}\nonumber\\
	 & -a_{111}^{(0)}a_{000}^{(1)}-a_{001}^{(0)}a_{110}^{(1)}
	-a_{010}^{(0)}a_{101}^{(1)}-a_{100}^{(0)}a_{011}^{(1)}\Big).\label{formula}
\end{align}
This completes the proof.
\end{proof}

\section{Optimizing the average tangle}

As a measure for pure bipartite entanglement we first take the tangle or the square of the I-concurrence (see Eq.(\ref{deftangle})). 
Now, in four qubits there are 4 bipartite cuts consisting of one-qubit verses the rest 3 quibts and  3 bi-partite cuts
consisting of 2-qubits verses the rest 2 qubits. Denoting the four qubits by A, B, C, and D, we define
\begin{align}
\tau_1 & \equiv\frac{1}{4}\left(\tau_{A(BCD)}+\tau_{B(ACD)}+\tau_{C(ABD)}+\tau_{D(ABC)}\right)\label{t1}\\
\tau_2 & \equiv\frac{1}{3}\left(\tau_{(AB)(CD)}+\tau_{(AC)(BD)}+\tau_{(AD)(BC)}\right)\;,\label{t2}
\end{align}
where $\tau_{A(BCD)}$, for example, is the tangle between qubit A and qubits B,C,D. Similarly,
$\tau_{(AB)(CD)}$, for example, is the tangle between qubits A,B and qubits C,D. 
Note that the maximum possible value for $\tau_1$ is $1$ and the maximum possible value for $\tau_2$ is $3/2$ 
(since a maximally entangled $4\times 4$ bipartite state has tangle $3/2$). 
However, it is argued now that no 4 qubit pure state can achieve this value for $\tau_2$.

In~\cite{GBS}, it has been shown that 
\begin{equation}
\tau_1\leq\tau_2\leq\frac{4}{3}\tau_1\;.
\label{bounds}
\end{equation}
Hence, since $\tau_1$ is bounded by $1$, it follows that $\tau_2\leq 4/3<3/2$. That is, there are no 4-qubit states for which all the 3 reduced density matrices, obtained by tracing out two qubits, are proportional to the identity. Moreover, from the inequality above,
it follows that for states with $\tau_2=4/3$, $\tau_1$ must be equal to $1$. In the following theorem, we characterize all states with
$\tau_1=1$. 

\begin{theorem}~\cite{Ver02,Ver03}\label{kn}
Let $|\psi\rangle\in\mathcal{H}_4\equiv\mathbb{C}^{2}\otimes\mathbb{C}^{2}\otimes\mathbb{C}^{2}\otimes\mathbb{C}^{2}$ be a 
normalized four qubit state. Then,
$$
\tau_1\left(|\psi\rangle\right)=1\;\;\;\text{if and only if}\;\;\;|\psi\rangle\in\mathcal{A}\;,
$$
up to local unitary transformation.
\end{theorem}

A weaker version of the theorem above has been first pointed out in~\cite{Ver02}. In~\cite{Ver02} the authors argued
that among all the states in $G\mathcal{A}$, only states in $\mathcal{A}$ have $\tau_1=1$. A year later in~\cite{Ver03}
theorem~\ref{kn} was fully proved. Nevertheless, for the purpose of completeness, we provide here a proof of 
Theorem~\ref{kn} for all states in $\mathcal{H}_4$, independently of the work in~\cite{Ver02,Ver03}.

\begin{proof}
Using the Kempf-Ness theorem~\cite{KN} (see also Appendix~\ref{KempfNess}) applied to $G$ (defined above)
and the fact that $G\mathcal{A}$ is the set of stable vectors~\footnote{A state $\psi$ is stable if the set of states 
$G\psi$ is closed in $\mathcal{H}_4$.}, one can show
that $\psi \in K\mathcal{A}$, with $K\equiv SU(2)\otimes SU(2)\otimes SU(2)\otimes SU(2)$, if
and only if~\cite{Kly02}
\begin{equation}\label{x}
\langle \psi|X|\psi\rangle= 0\;,
\end{equation}
for all $X$ in Lie$(G)$. 
Note that Lie$(G)$ acting on $\mathcal{H}_4$ is the direct sum of the Lie algebras of $SL(2,\mathbb{C})$
acting on one tensor factor. Now, let $\psi$ be a normalized state with $\tau_1(\psi)=1$. Therefore, we can write
$\psi=|0\rangle|\varphi_0\rangle + |1\rangle|\varphi_1\rangle$, with 
$\langle\varphi_0|\varphi_0\rangle = 1/2=\langle\varphi_1|\varphi_1\rangle$ and
$\langle\varphi_0|\varphi_1\rangle=0$. Low, let
$$
U\equiv\left(\begin{array}
[c]{cc}%
a & b \\
-b^* & a ^*
\end{array}\right)
$$
be a unitary matrix with $|a|^2 + |b|^2 =1$. Then, if $U_1$ is $U$ acting only in the first tensor
factor we have
$$
\langle\psi|U_1|\psi\rangle = \frac{a + a^*}{2}\;.
$$
Thus, the maximum value is attained at $U=I$. This implies that the condition
in Eq.(\ref{x}) is true for the part of the Lie algebra coming from the elements that
act only on the first factor. The argument for the other factors is the
same.
\end{proof}

From the Eq.~(\ref{bounds}) it follows that among all the states with $\tau_1=1$, we have
$$
1\leq\tau_2\leq\frac{4}{3}\;.
$$
It is interesting to note that the 4-qubit GHZ state gives the minimum possible value for $\tau_2$.
That is, it is the least entangled state among all the states with $\tau_1=1$. 
On the other hand, for all the 3 cluster states defined above, $\tau_2=4/3$. In fact, as we will see later,
the cluster states are the {\it only} states that achieve the maximal value for $\tau_2$ in such a way that 
2 of the terms (i.e. tangles) appearing in the definition of $\tau_2$ (see Eq.~(\ref{t2})) are equal to $3/2$ and one of the terms
equals to $1$.

From Theorem~\ref{kn} and the inequality~(\ref{bounds}), it follows that only states in $\mathcal{A}$ can maximize $\tau_2$.
From the following theorem it also follows that states that maximize $\tau_2$ must have zero 4-tangle.

\begin{theorem}\label{ub}
Let $\psi\in\mathcal{H}_4$ be a 4-qubits pure state and denote by $\tau_{ABCD}(\psi)$ its 4-tangle (defined above). Then,
$$
\tau_2(\psi)=\frac{4\tau_1(\psi)-\tau_{ABCD}(\psi)}{3}\;.
$$
\end{theorem}
\begin{remark}
The equation above can be written as $\tau_{ABCD}=4\tau_1-3\tau_2$, where $4\tau_1$ can be interpreted as the total amount
of entanglement in the system, whereas $3\tau_2$ can be interpreted as the total amount of entanglement shared among groups consisting of two qubits each. In this sense, the 4-tangle can be interpreted as the residual entanglement that can not be shared among the two qubits groups. Note that from the equation above it is obvious that the 4-tangle is invariant under permutations.
\end{remark}

\begin{proof}
Following the same notations as in Proposition~\ref{4tangle}, we
denote $|\psi\rangle=\sqrt{p_0}|0\rangle|\phi^0\rangle+\sqrt{p_1}|1\rangle|\phi^1\rangle$, where the 3-qubits states (qubits $BCD$), 
$|\phi^j\rangle\in\mathbb{C}^2\otimes\mathbb{C}^2\otimes\mathbb{C}^2$ with $j=0,1$, are orthonormal.
We also denote by $\mathfrak{D}_X$ with $X\in\{B,C,D\}$ the \emph{discriminant} of $|\psi\rangle$~\cite{GBS}:
$$
\mathfrak{D}_X\equiv\text{Tr}\left(\sigma_{X}^{00}\sigma_{X}^{11}-\sigma_{X}^{01}\sigma_{X}^{10}\right)\;,
$$
where
$$
\sigma_{X}^{kk'}\equiv\text{Tr}_{\neq X}|\phi^k\rangle\langle\phi^{k'}|\;\text{ with }\;k,k'\in\{0,1\}
$$
and the trace is taken over all the remaining two qubits that are \emph{not} the $X$ qubit.
The sum of the discriminants is denoted by $\mathfrak{D}\equiv \mathfrak{D}_{B}+\mathfrak{D}_{C}+\mathfrak{D}_{D}$.

With these notations the one-qubit reduced density matrices can be written as follows:
\begin{align}
	& \rho^A 
		= \text{Tr}_{\neq A}|\psi\rangle\langle\psi|
		=p_0|0\rangle\langle 0|+p_1|1\rangle\langle 1|, \nonumber\\
	& \rho^{X}=\text{Tr}_{\neq X}|\psi\rangle\langle\psi|
		=p_0\sigma_{X}^{00}+p_1\sigma_{X}^{11}.
\end{align}
Similarly, the two-qubit reduced density matrices are given by
\begin{equation}
\rho^{AX}=\text{Tr}_{\neq AX}|\psi\rangle\langle\psi|
	=\sum_{k,k'\in\{0,1\}}\sqrt{p_kp_{k'}}|k\rangle \langle k'|\otimes\sigma_{X}^{kk'}\;.
\end{equation}
Substituting these reduced densities matrices in the expressions for the linear entropy gives
$$
S_L(\rho^{AX})-S_L(\rho^{X})=4p_0p_1\mathfrak{D}_{X}\;.
$$
Since $S_{L}(\rho^{A})=4p_0p_1$, summing over $X$ gives
$$
4\tau_1-3\tau_2=4p_0p_1(1-\mathfrak{D})\;.
$$
Now, if we denote 
$$
|\phi^k\rangle=\sum_{i\in\{0,1\}^{3}}a^{(k)}_i|i\rangle
$$ 
for $k=0,1$, then a straightforward calculation (see Eq.(3.38) in~\cite{GBS}) gives
\begin{align}
	&\mathfrak{D}  =1-\Big|a_{000}^{(0)}a_{111}^{(1)}+a_{110}^{(0)}a_{001}^{(1)}
	+a_{101}^{(0)}a_{010}^{(1)}+a_{011}^{(0)}a_{100}^{(1)}\nonumber \\
	 &-a_{111}^{(0)}a_{000}^{(1)}-a_{001}^{(0)}a_{110}^{(1)}
	-a_{010}^{(0)}a_{101}^{(1)}-a_{100}^{(0)}a_{011}^{(1)}\Big|^2.
	\end{align}
A comparison of this expression with the one given in Eq.(\ref{formula}) implies that $4\tau_1-3\tau_2=\tau_{ABCD}$.	
\end{proof}

From Theorem~\ref{kn} and Theorem~\ref{ub}, we have the following corollary.
\begin{corollary}
A normalized state $|\psi\rangle\in\mathcal{H}_4$ is maximally entangled (i.e. $\tau_2(|\psi\rangle)=4/3$)
if and only if up to local unitary $|\psi\rangle\in\mathcal{M}$, where $\mathcal{M}$ is the set of states in $\mathcal{A}$ 
with zero 4-tangle.
\end{corollary}

A state $\psi=\sum_{j=0}^{3}z_ju_j$ in $A$ depends on four complex parameters $z_j$ $(j=0,1,2,3)$. The condition that the 4-tangle 
$\tau_{ABCD}(\psi)=|\sum_{j=0}^{3}z_{j}^{2}|^2=0$ implies that 
the states in the maximally entangled class $\mathcal{M}$ are characterized by 4 \emph{real} parameters since we also have the normalization condition and we ignore the global phase. If we write $z_j=\sqrt{p_j}e^{i\theta_j}$ in its polar form (with non-negative
$p_j$ and $\theta_j\in[0,2\pi]$)
we can characterize the class 
$\mathcal{M}$ as follows:
\begin{equation}
\mathcal{M}=\left\{\sum_{j=0}^{3}\sqrt{p_j}e^{i\theta_j}u_j\Big|\;\sum_{j=0}^{3}p_j=1\;,\; \sum_{j=0}^{3}p_je^{2i\theta_j}=0\right\}\;.
\label{classm}
\end{equation} 

\subsection*{Minimization of $\tau_2$}

From theorem~\ref{ub}, it also follows that the states in $\mathcal{A}$ with the minimum possible value $\tau_2=1$,
can be characterized as follows. Denote by $\mathcal{T}_{\min}$ the class of all such states. Then,
\begin{align}
\mathcal{T}_{\min} & \equiv\left\{\psi\in\mathcal{A}\Big|\tau_2(\psi)=1\right\}\nonumber\\
& =\left\{\sum_{j=0}^{3}x_ju_j\Big|\;\sum_{j=0}^{3}x_{j}^{2}=1\;,\;x_j\in\mathbb{R}\right\}
\end{align}
Note that the four qubits GHZ state belongs to $\mathcal{T}_{\min}$. In this sense, the GHZ state is a state in 
$\mathcal{A}$ with the least amount of entanglement.

\section{Different measures of entanglement}

Up to now we took the measure in~Eqs.(\ref{t1},\ref{t2}) to be the tangle, which is given in terms of the linear entropy.
The measures of entanglement that we consider in this section are Renyi entropy of entanglement and Tsallis entropy of entanglement. We denote these measures by $E_{\text{R}}^{(\alpha)}$ and $E_{\text{T}}^{(\alpha)}$, respectively.
These bipartite measures of entanglement are defined as follows.
Given a bipartite state $|\psi^{AB}\rangle\in\mathbb{C}^{n}\otimes\mathbb{C}^{m}$,  the entanglements 
$E^{(\alpha)}_{\text{R}}$
and $E^{(\alpha)}_{\text{T}}$,
measured by the Renyi and Tsallis entropies, are given by
\begin{align}
& E^{(\alpha)}_{\text{R}}\left(|\psi^{AB}\rangle\right)\equiv \frac{1}{1-\alpha}\log\text{Tr}\rho_{r}^{\alpha}\nonumber\\
& E^{(\alpha)}_{\text{T}}\left(|\psi^{AB}\rangle\right)\equiv \frac{1}{1-\alpha}\left(\text{Tr}\rho_{r}^{\alpha}-1\right)\nonumber
\end{align}
where $\rho_{r}=\text{Tr}_{B}|\psi^{AB}\rangle\langle\psi^{AB}|$ is the reduced density matrix
and the log is base 2. Note that both Renyi and Tsallis entropies approach the von-Neumann entropy 
in the limit $\alpha\rightarrow 1$. The Tsallis entropy is concave (see for example~\cite{Hu06}) 
and therefore $E^{(\alpha)}_{\text{T}}$ is an ensemble entanglement monotone 
(i.e. non-increasing on average under LOCC). 
The Renyi entropy is also concave for $0<\alpha\leq 1$, but only Shur concave for $\alpha>1$~\cite{Da02}.
Hence, for $\alpha>1$, $ E^{(\alpha)}_{\text{R}}$ is only a deterministic monotone 
(i.e. non-increasing under deterministic LOCC). 
Nevertheless, unlike the Tsallis $\alpha$-entropy of entanglement with $\alpha\neq 1$, 
the Renyi $\alpha$-entropy of entanglement is normalized nicely
so that it is equal $\log d$ for maximally entangled states in $\mathbb{C}^{d}\otimes\mathbb{C}^d$.
 
Similar to the definition of the average tangles~(\ref{t1},\ref{t2}), we define the average $\alpha-$entropy of entanglement of four qubits states as follows:
\begin{align}
E_{1}^{(\alpha)} & \equiv\frac{1}{4}\left(E_{A(BCD)}^{(\alpha)}+E_{B(ACD)}^{(\alpha)}+E_{C(ABD)}^{(\alpha)}
+E_{D(ABC)}^{(\alpha)}\right)\\
E_{2}^{(\alpha)} & \equiv\frac{1}{3}\left(E_{(AB)(CD)}^{(\alpha)}+E_{(AC)(BD)}^{(\alpha)}+E_{(AD)(BC)}^{(\alpha)}\right),
\label{def}
\end{align}
where $E_{A(BCD)}^{(\alpha)}$, for example, is the Tsallis or Renyi $\alpha$-entropy of entanglement between qubit A and qubits B,C,D, where it will be clear from the context if we mean Renyi or Tsallis. 
Note that due to Eq.(\ref{bounds}) the maximum value of $E_{2}^{(\alpha)}$ cannot be 2 ebits
(i.e. the same value as the the value for two bell states).  

From Theorem~\ref{kn} we know that
$E_{1}^{(\alpha)}(|\psi\rangle)=1$ iff $|\psi\rangle\in\mathcal{A}$. However, for the general measures of entanglement we do not have an equation analog to Eq.~(\ref{bounds})
and therefore, can not argue that if $|\psi\rangle$ maximize $E_{2}^{(\alpha)}$ then it must maximize 
$E_{1}^{(\alpha)}$. 
Nevertheless, we will see that for the average $\alpha-$entropy of entanglement with $\alpha\geq 2$ 
this is indeed the case~\footnote{We are willing to conjecture that it is true for all $\alpha\geq 0$.}.

In order to optimize $E_{2}^{(\alpha)}$, we first prove the following theorem.

\begin{theorem}\label{lm}
Let $\rho$ be a $4\times 4$ normalized density matrix with eigenvalues $\lambda_0\geq\lambda_1\geq\lambda_2\geq\lambda_3$.
Let $S_L(\rho)=\tau$ and denote $x(\tau)=\sqrt{1-\frac{2}{3}\tau}$ and $y(\tau)=\sqrt{1-\frac{3}{4}\tau}$. 
The maximum (minimum)
possible value of either Tsallis or Renyi entropies with $0<\alpha< 2$ ($\alpha>2$) is obtained if and only if 
$\{\lambda_i\}$ is given by
\begin{equation}\label{lm1}
\left\{\frac{1+3x(\tau)}{4}\;,\;\frac{1-x(\tau)}{4}\;,\;\frac{1-x(\tau)}{4}\;,\;\frac{1-x(\tau)}{4}\right\}
\end{equation}
The minimum (maximum) possible value of Tsallis or Renyi entropies with $0<\alpha< 2$ ($\alpha>2$) is obtained if and only if the set $\{\lambda_i\}$ is given by
\begin{align}
& \left\{\frac{1+x(\tau)}{4},\frac{1+x(\tau)}{4},\frac{1+x(\tau)}{4},\frac{1-3x(\tau)}{4}\right\}\;(\frac{4}{3}\leq \tau\leq \frac{3}{2})\nonumber\\
& \left\{\frac{1+y(\tau)}{3},\frac{1+y(\tau)}{3},\frac{1-2y(\tau)}{3},0\right\}\;\;\;\;\;\;\;\;\;\;\;\;\;\;(1\leq \tau\leq \frac{4}{3})
\nonumber\\
& \left\{\frac{1+\sqrt{1-\tau}}{2},\frac{1-\sqrt{1-\tau}}{2},0,0\right\}\;\;\;\;\;\;\;\;\;\;\;\;\;\;\;\;\;\;\;(0\leq \tau\leq 1)\label{L1}
\end{align}
\end{theorem}
One direction of theorem~\ref{lm} has been proven in~\cite{Ber03}.  To complete the proof of theorem~\ref{lm},
we show in Appendix \ref{aaa123} that the Renyi and Tsillsa entropies obtain there extremum values \emph{only} 
for the sets of eigenvalues that appear in the theorem.

\subsection{The state $|L\rangle$}

In this section we show that the state $|L\rangle$ in Eq.~(\ref{statel}) has the remarkable property that it \emph{maximizes} the average Tsallis entropy of entanglement $E^{(\alpha)}_{2}$ for \emph{all} $\alpha\geq 2$.
In addition, among all the states in $\mathcal{M}$, the state $|L\rangle$ \emph{minimizes}  $E^{(\alpha)}_{2}$ for \emph{all} $0\leq\alpha\leq 2$. 

\begin{theorem}\label{L}
{\rm \textbf{(a)}}
Let $\psi\in\mathcal{H}_4$. Then,
\begin{align}
E_{2}^{(\alpha)}(\psi)\leq E_{2}^{(\alpha)}(|L\rangle)\text{  for  all } \alpha> 2\nonumber 
\end{align}
with equality if and only if up to local unitaries  $\psi=|L\rangle$.\\
{\rm \textbf{(b)}} Let $\psi\in\mathcal{M}$. Then,
\begin{align}
E_{2}^{(\alpha)}(\psi)\geq E_{2}^{(\alpha)}(|L\rangle)\text{  for  all } 0<\alpha< 2\nonumber 
\end{align}
with equality if and only if up to local unitaries  $\psi=|L\rangle$.
\end{theorem}

\begin{proof}
Let 
$\tilde{E}_{\max}^{(\alpha)}(t)$ and 
$\tilde{E}_{\text{min}}^{(\alpha)}(t)$ be, respectively, the maximum and minimum values $E^{(\alpha)}_{\text{T}}(\psi)$ can take among all normalized bipartite states $\psi\in\mathbb{C}^4\otimes\mathbb{C}^4$ with tangle $\tau(\psi)=t$.
Further, let $f_{\max}^{(\alpha)}(t_1,t_2,t_3)\equiv1/3\sum_{k=1}^{3}\tilde{E}_{\max}^{(\alpha)}(t_k)$ and 
$f_{\min}^{(\alpha)}(t_1,t_2,t_3)\equiv1/3\sum_{k=1}^{3}\tilde{E}_{\min}^{(\alpha)}(t_k)$.
Thus, for a 4 qubit normalized state $|\psi\rangle\in\mathcal{H}_4$, with $\tau_{(12)(34)}=t_1$, $\tau_{(13)(24)}=t_2$, and $\tau_{(14)(34)}=t_3$, 
\begin{equation}\label{bd}
f_{\min}^{(\alpha)}(t_1,t_2,t_3)\leq E_{2}^{(\alpha)}(|\psi\rangle)\leq f_{\max}^{(\alpha)}(t_1,t_2,t_3)\;.
\end{equation}
From theorem~\ref{lm} and the definition of Tsallis entropy it follows that for $\alpha>2$
\begin{align}
& \tilde{E}_{\max}^{(\alpha)}(t) =\frac{1}{\alpha-1}\times\nonumber\\
& \left\{ 
                  \begin{array}{rll}
			\left[1-\frac{3(1+x)^\alpha}{4^\alpha}-\frac{(1-3x)^\alpha}{4^\alpha}\right]&\text{for } \frac{4}{3}\leq t\leq\frac{3}{2} \\
			\\
			\left[1-\frac{2(1+y)^\alpha}{3^\alpha}-\frac{(1-2y)^\alpha}{3^\alpha}\right]&\text{for } 1\leq t\leq\frac{4}{3}			
		\end{array}\right.
		\label{mx}
\end{align}
where $x=x(t)\equiv\sqrt{1-\frac{2}{3}t}$ and $y=y(t)\equiv\sqrt{1-\frac{3}{4}t}$. Furthermore, from theorem~\ref{lm}
it follows that for $0<\alpha <2$, $\tilde{E}_{\min}^{(\alpha)}(t)$ is given by the exact same expression as in
Eq.(\ref{mx}).

Note that $E_{\max}^{(\alpha)}(t)$ and $E_{\min}^{(\alpha)}(t)$ are continuous at the point $t=4/3$ 
(i.e. $x=1/3$ and $y=0$), but not their derivatives.
Nevertheless, since the derivatives of $E_{\max}^{(\alpha)}(t)$ and $E_{\min}^{(\alpha)}(t)$ in both regions $1< t<4/3$ and $3/4<t<3/2$ are positive, it follows that $E_{\max}^{(\alpha)}(t)$ and $E_{\min}^{(\alpha)}(t)$ are both monotonically increasing with $t$.
Now, from Eq.~(\ref{bounds}) it follows that $t_1+t_2+t_3\leq 4$ and therefore the RHS of the Eq.~(\ref{bd}) reach its maximum value when  $t_1+t_2+t_3= 4$. 
Hence, for any $\psi\in\mathcal{H}_4$, $E_{2}^{(\alpha)}(\psi)$ is bounded above by ($\alpha>2$)
$$
U^{(\alpha)}\equiv\max\Big\{f_{\max}^{(\alpha)}(t_1,t_2,t_3)\;\big|\;\sum_{k=1}^{3}t_k=4\;,\;t_k\leq 3/2\Big\}\;.
$$
Similarly, for any $\psi\in\mathcal{M}$, $E_2(\psi)$ is bounded below by ($0<\alpha<2$)
$$
L^{(\alpha)}\equiv\min\Big\{f_{\min}^{(\alpha)}(t_1,t_2,t_3)\;\big|\;\sum_{k=1}^{3}t_k=4\;,\;t_k\leq 3/2\Big\}\;.
$$
In Appendix~\ref{optimization} we show that $U^{(\alpha)}=f_{\max}^{(\alpha)}(4/3,4/3,4/3)$
and $L^{(\alpha)}=f_{\min}^{(\alpha)}(4/3,4/3,4/3)$. We also show that the point $t_1=t_2=t_3=4/3$
is the \emph{only} point of global max for $f_{\max}^{(\alpha)}$ and the only point of global min for 
$f_{\min}^{(\alpha)}$. 

Let $|\psi\rangle\in\mathcal{H}_4$, and denote by $\{P_j\}$, $\{Q_j\}$, and $\{R_j\}$ the eigenvalues of the reduced density matrices of $|\psi\rangle$ obtained after tracing out qubits C and D, B and D, and B and C, respectively.
From the analysis above and the results in appendix~\ref{optimization}, $E_{2}^{(\alpha)}(|\psi\rangle)=U^{(\alpha)}$
if and only if $\tau_{(12)(34)}(\psi)=\tau_{(12)(34)}(\psi)=\tau_{(12)(34)}(\psi)=4/3$ and the distributions
$\{P_j\}$, $\{Q_j\}$, and $\{R_j\}$ are all of the form given in Eq.(\ref{L1}) with $\tau=4/3$. That is, for $\alpha>2$
$E_{2}^{(\alpha)}(|\psi\rangle)=U^{(\alpha)}$ if and only if $\{P_j\}=\{Q_j\}=\{R_j\}=\{1/3,1/3,1/3,0\}$. Similarly,
if $|\psi\rangle\in\mathcal{M}$ then, for $0<\alpha<2$,  $E_{2}^{(\alpha)}(|\psi\rangle)=L^{(\alpha)}$
if and only if $\{P_j\}=\{Q_j\}=\{R_j\}=\{1/3,1/3,1/3,0\}$.

A priori it is not clear that such a state with  $\{P_j\}=\{Q_j\}=\{R_j\}=\{1/3,1/3,1/3,0\}$ exists in $\mathcal{H}_4$.
However, now we show that up to local unitaries there exists exactly one state with this property and the state
is $|L\rangle$.

First note that if $\{P_j\}=\{Q_j\}=\{R_j\}=\{1/3,1/3,1/3,0\}$ then up to local unitaries $|\psi\rangle\in\mathcal{M}\subset\mathcal{A}$.
Therefore, we can write $|\psi\rangle=z_0u_0+z_1u_1+z_2u_2+z_3u_3$. In this case, 
the eigenvalues of the reduced density matrices are given by ($j=0,1,2,3$):
\begin{equation}
P_j=|z_j|^2\;\;,\;\;Q_j=\left|\sum_{k=0}^{3}A_{jk}z_k \right|^2\;\;,\;\;R_j=\left|\sum_{k=0}^{3}B_{jk}z_k \right|^2\; ,
\label{dis}
\end{equation}
where $A_{jk}$ and $B_{jk}$ are the matrix element of the two orthogonal $4\times 4$ orthogonal matrices:
$$
A=\frac{1}{2}\begin{pmatrix}
			1 & 1 & 1 & 1 \\ 
			1 & 1 & -1 & -1 \\ 
			1 & -1 & 1 & -1 \\ 
			1 & -1 & -1 & 1 
			
		\end{pmatrix}
		\text{   and   }
		B=\frac{1}{2}\begin{pmatrix}
			1 & -1 & -1 & -1 \\ 
			1 & -1 & 1 & 1 \\ 
			1 & 1 & -1 & 1 \\ 
			1 & 1 & 1 & -1 
			
		\end{pmatrix}.
$$
Now, since $\{P_j\}=\{1/3,1/3,1/3,0\}$ we have 
$$
z_0=0 \;\;,\;\;z_1=\frac{1}{\sqrt{3}}e^{i\theta_1}\;\;,\;\;z_2=\frac{1}{\sqrt{3}}e^{i\theta_2}\;\;,\;\;z_3=\frac{1}{\sqrt{3}}e^{i\theta_3}\;,
$$
were we have used the fact that the transformation $|\psi\rangle\rightarrow|\psi'\rangle=\sum_{j=0}^{3}z_{\sigma(j)}u_j$ can be achieved by local unitaries for all permutations $\sigma$. With these values for $z_j$ we get
\begin{align}
& Q_0=R_0=\frac{1}{12}\left|e^{i\theta_1}+e^{i\theta_2}+e^{i\theta_3}\right|^2\nonumber\\
& Q_1=R_1=\frac{1}{12}\left|-e^{i\theta_1}+e^{i\theta_2}+e^{i\theta_3}\right|^2\nonumber\\
& Q_2=R_2=\frac{1}{12}\left|e^{i\theta_1}-e^{i\theta_2}+e^{i\theta_3}\right|^2\nonumber\\
& Q_3=R_3=\frac{1}{12}\left|e^{i\theta_1}+e^{i\theta_2}-e^{i\theta_3}\right|^2\;.
\end{align}
Thus, $\{Q_j\}=\{R_j\}=\{1/3,1/3,1/3,0\}$ if and only if up to permutation $e^{i\theta_1}=\pm1$,
$e^{i\theta_2}=\pm e^{i\pi/3}$, and $e^{i\theta_3}=\pm e^{i2\pi/3}$. From proposition~\ref{localunitary1} it follows 
that up to local unitaries $|\psi\rangle=|L\rangle$.
\end{proof}

\subsection{The State $|M\rangle$}

In this section we show that the state $|M\rangle$ in Eq.~(\ref{entropystates}) has the remarkable property that it \emph{maximizes} the average Tsallis entropy of entanglement $E^{(\alpha)}_{2}$ for \emph{all} $0\leq\alpha\leq 2$.
In addition, among all the states in $\mathcal{M}$, the state $|M\rangle$ \emph{minimizes}  $E^{(\alpha)}_{2}$ for \emph{all} $\alpha\geq 2$. 

\begin{theorem}\label{symmetric}
\textbf{(a)}
Let $\psi\in\mathcal{A}$. Then,
\begin{equation}
E_{2}^{(\alpha)}(\psi)\leq E_{2}^{(\alpha)}(|M\rangle)\text{  for  } 0<\alpha< 2 
\end{equation}
with equality if and only if up to local unitaries $\psi=|M\rangle$.\\
\textbf{(b)} Let $\psi\in\mathcal{M}$. Then,
\begin{equation}
E_{2}^{(\alpha)}(\psi)\geq E_{2}^{(\alpha)}(|M\rangle)\text{  for  } \alpha>2 
\end{equation}
with equality if and only if up to local unitaries $\psi=|M\rangle$.
\end{theorem}
\begin{proof}
Let $\psi\in\mathcal{A}$, and
denote $z_j\equiv\sqrt{p_j}e^{i\theta_j}$ for $j=0,1,2,3$. With this notations, $\mathcal{E}_1(\psi)=\sum_{j=0}^{3}p_je^{i2\theta_j}$,
and for a fixed value of $|\mathcal{E}_1|\equiv a$, $(a\geq 0)$, the formula in Theorem~\ref{ub}
can be written as
$$
t_1(\psi)+t_2(\psi)+t_3(\psi)=4-a^2\;.
$$
Moreover, note that if $|\mathcal{E}_1(\psi)|= a$ then $p_j\leq(1+a)/2$ for all $j=0,1,2,3$.
Therefore, we denote by 
$\tilde{E}_{\max}^{(\alpha)}(a,t)$, the maximum value $E^{(\alpha)}_{\text{T}}(\varphi)$
can take among all normalized bipartite states $\varphi\in\mathbb{C}^4\otimes\mathbb{C}^4$ 
with tangle $\tau(\varphi)=t$ and Schmidt coefficients $p_j\leq(1+a)/2$. 

Now, a simple calculation shows that for 
a distribution of the form $p_0\geq p_1=p_2=p_3$, we get that $p_j\leq(1+a)/2$ if and only if  
$t\geq 2(1-a)(2+a)/3$. Therefore, from theorem~\ref{lm} it follows that this is the optimal distribution
for $t\geq 2(1-a)(2+a)/3$. On the other hand, if $t< 2(1-a)(2+a)/3$, the optimal distribution (up to permutation) is given by
$p_0=\frac{1+a}{2}\geq p_1\geq p_2=p_3$. This is follows from the extension of the results in~\cite{Ber03}, and in 
particular~\cite{Dominic}, it is a consequence of Eq.(22) in~\cite{Ber03}. Hence, we conclude that for $0<\alpha<2$

\begin{align}\label{mxx}
&\tilde{E}_{\max}^{(\alpha)}(a,t) =\frac{1}{\alpha-1}\nonumber\\
&\left\{ 
                  \begin{array}{rll}
			&\left[1-\frac{3(1-x)^\alpha}{4^\alpha}-\frac{(1+3x)^\alpha}{4^\alpha}\right]\;\text{ for }t
			\geq\frac{2(1-a)(2+a)}{3} \\
			\\
			&
			\left[1-\left(\frac{1+a}{2}\right)^\alpha-\frac{(1-a+2\omega)^\alpha}{6^\alpha}-\frac{2(1-a-\omega)^\alpha}{6^\alpha}\right]\;\text{otherwise}			
		\end{array}\right.
\end{align}
where $x=x(t)\equiv\sqrt{1-\frac{2}{3}t}$ and $\omega\equiv\omega(t)\equiv\sqrt{4-2a-2a^2-3t}$. 

Similar to the definition in theorem~\ref{L}, we define 
\begin{align}
& f_{\max}^{(\alpha)}(a,t_1,t_2,t_3)
\equiv\nonumber\\
&\frac{1}{3}\left(\tilde{E}_{\max}^{(\alpha)}(a,t_1)+\tilde{E}_{\max}^{(\alpha)}(a,t_2)+\tilde{E}_{\max}^{(\alpha)}(a,t_3)\right)\;.
\end{align}
We therefore have for $0<\alpha<2$
$$
E_{2}^{(\alpha)}(\psi)\leq f_{\max}^{(\alpha)}(a,t_1,t_2,t_3).
$$
A straightforward calculation, similar to the one given in Appendix~\ref{optimization}, shows that
the global maximum of the function
$f_{\max}^{(\alpha)}(a,t_1,t_2,t_3)$ is unique and is obtained at the point $a=0$ (i.e. $\psi\in\mathcal{M}$) and $t_1=t_2=t_3=4/3$.
Therefore, from theorem~\ref{lm} (and in particular from Eq.~(\ref{lm1}) with $\tau=4/3$), it follows that this global maximum is obtained if and only if
\begin{equation}\label{pqr}
P_j,\;Q_j,\;R_j\in\Big\{\frac{1}{2},\frac{1}{6},\frac{1}{6},\frac{1}{6}\Big\}\;\;\;\forall\;\;j=0,1,2,3\;,
\end{equation}
where the sets $\{P_j\}$, $\{Q_j\}$, and $\{R_j\}$ are the eigenvalues of the reduced density matrices of $\psi$ 
obtained after tracing out qubits C and D, B and D, and B and C, respectively.
A priori, it is not clear that a four qubits state with these properties exists. We show now that there exists only
one state (up to local unitaries) with these properties, and it is the state $|M\rangle$.
Note also that if such a $\psi$ exists then $\psi\in\mathcal{M}\subset\mathcal{A}$.

Up to a local unitary (see proposition~\ref{localunitary}), w.l.o.g. we can assume that $P_0=1/2$ and 
$P_1=P_2=P_3=1/6$. Further, due to the freedom of global phase, we have
$$
z_0=\frac{1}{\sqrt{2}},\;z_1=\frac{1}{\sqrt{6}}e^{i\theta_1},\;z_2=\frac{1}{\sqrt{6}}e^{i\theta_2},\;z_3=\frac{1}{\sqrt{6}}e^{i\theta_3}\;.
$$
Next, the condition that three of the $Q_j$s and three of the $R_j$s equal to 1/6 implies that $|\psi\rangle$ equals
to one of the 8 states:
$$
\frac{1}{\sqrt{2}}u_0\pm\frac{i}{\sqrt{6}}u_1\pm\frac{i}{\sqrt{6}}u_2\pm\frac{i}{\sqrt{6}}u_3
$$ 
Using proposition~\ref{localunitary} we get that all these eight states are equivalent under local unitaries.
We are therefore left with one state
\begin{align}
\psi=\frac{1}{\sqrt{2}}u_0+\frac{i}{\sqrt{6}}u_1+\frac{i}{\sqrt{6}}u_2+\frac{i}{\sqrt{6}}u_3
\end{align}
Up to a local unitary, this state is the same as $|M\rangle$ in Eq.~(\ref{entropystates}).

The proof of part \textbf{(b)} of the theorem follows the same lines as above with $a=0$ and $\alpha>2$. 
\end{proof}

\subsection{The cluster states}

We are now ready to analyze the maximally entangled states in $\mathcal{M}$ for which two of the tangles $\{\tau_{(AB)(CD)},\tau_{(AC)(BD)},\tau_{(AD)(BC)}\}$ equal
to the maximal value of $3/2$ while the remaining tangle equals to $1$. 

\begin{theorem}\label{CS}
Let $|\psi\rangle\in\mathcal{A}$, and denote by $\{P_j\}$, $\{Q_j\}$, and $\{R_j\}$ the eigenvalues of the reduced density matrices
of $|\psi\rangle$ obtained after tracing out qubits C and D, B and D, and B and C, respectively. 
If the Shannon entropies $H\left(\{P_i\}\right)=H\left(\{Q_i\}\right)=2$, 
then $H\left(\{R_i\}\right)\leq 1$, with equality if and only if up to a local unitary $|\psi\rangle=|C_2\rangle\in\mathcal{M}$.
\end{theorem}
Note that from the theorem above it follows that, up to a local unitary, $|\psi\rangle=|C_1\rangle$ ( $|\psi\rangle=|C_3\rangle$) if one replace the condition
$H\left(\{P_i\}\right)=H\left(\{Q_i\}\right)=2$ with $H\left(\{Q_i\}\right)=H\left(\{R_i\}\right)=2$ ($H\left(\{P_i\}\right)=H\left(\{R_i\}\right)=2$).

\begin{proof}
Let $|\psi\rangle\in\mathcal{A}$, and define two orthogonal $4\times 4$ orthogonal matrices:
$$
A=\frac{1}{2}\begin{pmatrix}
			1 & 1 & 1 & 1 \\ 
			1 & 1 & -1 & -1 \\ 
			1 & -1 & 1 & -1 \\ 
			1 & -1 & -1 & 1 
			
		\end{pmatrix}
		\text{   and   }
		B=\frac{1}{2}\begin{pmatrix}
			1 & 1 & 1 & -1 \\ 
			1 & 1 & -1 & 1 \\ 
			1 & -1 & 1 & 1 \\ 
			1 & -1 & -1 & -1 
			
		\end{pmatrix}.
$$

The eigenvalues of the reduced density matrices are given by ($j=0,1,2,3$):
$$
P_j=|z_j|^2\;\;,\;\;Q_j=\left|\sum_{k=0}^{3}A_{jk}z_k \right|^2\;\;,\;\;R_j=\left|\sum_{k=0}^{3}B_{jk}z_k \right|^2\; .
$$
Hence, the equality $H\left(\{P_i\}\right)=2$ leads to 
$$
z_j=\frac{1}{2}e^{i\theta_j}\;\;,\;j=0,1,2,3\;\;.
$$
Now, since a quantum state is defined up to a global phase, w.l.o.g. we can take $\theta_0=-\theta_1$ which reduces the number of free parameters to 3.
Next, the condition $H(\{Q_j\})=2$ together with proposition 2 implies that up to a local unitary
$$
|\psi\rangle=\frac{1}{2}u_0-\frac{1}{2}u_1+\frac{e^{i\gamma}}{2}u_2+\frac{e^{i\gamma}}{2}u_3\;,
$$
where $\gamma$ is the only free real parameter left. 
Now, it is a simple calculation to check that at least two of the $R_j$ ($j=0,1,2,3$)
equals to zero. Therefore, $H\left(\{R_j\}\right)\leq 1$. If we require that $H\left(\{R_j\}\right)= 1$, then $e^{i\gamma}=\pm i$
and up to a local unitary $|\psi\rangle=|C_2\rangle$.
\end{proof}

Next we show that the three cluster states are the only states (up to local unitaries) that maximize
the average Renyi entropy of degree $\alpha\geq 2$. 

\subsubsection{Cluster states are the only states that maximize the average Renyi entropy with $\alpha\geq 2$}

\begin{theorem}
Let $\psi\in\mathcal{H}_4$ and $\alpha\geq 2$. Then,
$$
E_{2}^{(\alpha)}(\psi)\leq 5/3\;,
$$
with equality if and only if, up to local unitaries, $\psi$ is one of the cluster states given in Eqs.(\ref{C1},\ref{C2},\ref{C3}).
\end{theorem}
\begin{proof}
We first prove it for the case $\alpha=2$. In this case, the Renyi entropy of degree 2 (which also called collision entropy)
can be expressed in terms of the linear entropy. This implies that the Renyi entropy of entanglement $E^{(\alpha=2)}_{R}$ can be expressed in terms of the tangle:
$$
E^{(\alpha=2)}_{R}(|\psi^{AB}\rangle)=-\log\left(1-\frac{1}{2}\tau(|\psi^{AB}\rangle)\right)\;.
$$
To simplify notations, we denote by $t_1$, $t_2$ and $t_3$ the values of $\tau_{(AB)(CD)}(\psi)$,
$\tau_{(AC)(BD)}(\psi)$, and $\tau_{(AD)(BC)}(\psi)$, respectively. With these notations we have
$$
E_{2}^{(\alpha=2)}(\psi)=-\frac{1}{3}\log\left[\left(1-\frac{1}{2}t_1\right)\left(1-\frac{1}{2}t_2\right)\left(1-\frac{1}{2}t_3\right)\right]
$$
Note that since $\tau_2\leq 4/3$ we have $t_1+t_2+t_3\leq 4$. Now, the function $-\log(1-t/2)$ increase with $t$.
Therefore, $E_{2}^{\alpha=2}$ obtains its maximum value when $t_1+t_2+t_3=4$. 
Now, we define 
$$
f(t_1,t_2)\equiv\left(1-\frac{1}{2}t_1\right)\left(1-\frac{1}{2}t_2\right)\left(\frac{t_1+t_2}{2}-1\right)
$$ 
on the domain 
$$
D=\{(t_1,t_2)\big| 1\leq t_1\leq 3/2\;,\;5/2-t_1 \leq t_2\leq 3/2\}.
$$
A simple analysis of the function $f(t_1,t_2)$ implies that $f(t_1,t_2)$ obtains its minimum value of $1/32$
only at the points $(1,3/2)$, $(3/2,1)$, and $(3/2,3/2)$. Therefore, for states with these values for $t_1$ $t_2$ and $t_3=4-t_1-t_2$, $E_{2}^{(\alpha=2)}$ obtains its maximum value. In theorem~\ref{CS} we have seen that 
the 3 cluster states are the only ones with these values of $t_1$, $t_2$, and $t_3$. This completes the proof of the theorem for $\alpha=2$. The case for $\alpha>2$ follows immediately from the fact that the Renyi entropy is 
a non-increasing function of $\alpha$ and therefore $E_{2}^{(\alpha=2)}\geq E_{2}^{(\alpha>2)}$. To complete the proof, we observe that for the cluster states $E_{2}^{(\alpha)}=5/3$ for all $\alpha$.
\end{proof}

\section{More maximally entangled 4-qubits states}

In this section we characterize all the states in $\psi\in\mathcal{A}$ that maximize the average $\alpha$-entropy of 
entanglement $E^{(\alpha)}_2$ (as defined in Eq.(\ref{def})), for given values of $\tau_{(AB)(CD)}$, $\tau_{(AC)(BD)}$, and $\tau_{(AD)(BC)}$. We also characterize the states in $\mathcal{M}$ for which one of the three tangles $\tau_{(AB)(CD)}$, 
$\tau_{(AC)(BD)}$, and $\tau_{(AD)(BC)}$, is equal to $3/2$.

We start with a definition of a class $\mathcal{C}$ of maximally entangled states:
\begin{align}
& \mathcal{C}\equiv\nonumber\\
& \left\{\sqrt{p}e^{i\theta}u_0+\sqrt{\frac{1-p}{3}}\sum_{j=1}^{3}u_j\Big|\;\frac{1}{2}\leq p\leq 1
,\;\cos ^2\theta\leq\frac{1-p}{3p}\right\}\nonumber
\end{align}
\begin{proposition}
Up to local unitaries, the class $\mathcal{C}$ consists of all the states in~$\mathcal{A}$ such that all three distributions
$\{P_j\}$, $\{Q_j\}$, and $\{R_j\}$, as defined in Eq.~(\ref{dis}), have the form $\{p_0,p_1,p_1,p_1\}$
with $p_0>p_1$; i.e. they are of the form given in Eq.(\ref{lm1}). 
\end{proposition}
\begin{remark}
Note that the state $|M\rangle$ belongs to $\mathcal{C}$. It corresponds to $\theta=\pi/2$ and $p=1/2$.
\end{remark}
\begin{proof}
Let $\psi=z_0u_0+z_1u_1+z_2u_2+z_3u_3$ be a state in $\mathcal{A}$ with the properties mentioned in the proposition. Therefore, since $P_j=|z_j|^2$, up to local unitaries
$$
\psi=\sqrt{p}u_0+\sqrt{\frac{1-p}{3}}\sum_{j=1}^{3}e^{i\theta_j}u_j\;,
$$
with $p\geq 1/2$.
From the definitions of $\{Q_j\}$ and $\{R_j\}$ in Eq.~(\ref{dis}), and from the requirement that three of the $Q_j$s
are equal and smaller than $1/6$ and also three of the $R_j$s are equal and smaller than $1/6$, we get that
$\theta_1=\theta_2=\theta_3\equiv\theta$, and  $\cos ^2\theta\leq (1-p)/3p$. Therefore, up to a global phase
$\psi\in\mathcal{C}$.
\end{proof}

\begin{corollary}
Let $\psi\in\mathcal{H}_4$ and $\phi\in\mathcal{C}$. If 
\begin{align}
& \tau_{(AB)(CD)}(\psi)\leq\tau_{(AB)(CD)}(\phi)\nonumber\\ 
& \tau_{(AC)(BD)}(\psi)\leq\tau_{(AC)(BD)}(\phi)\nonumber\\
& \tau_{(AD)(BC)}(\psi)\leq\tau_{(AD)(BC)}(\phi)\;,\nonumber
\end{align}
then
\begin{align}
& E_{2}^{(\alpha)}(\psi)\leq E_{2}^{(\alpha)}(\phi)\;\;\text{for  }0<\alpha<2\nonumber\\
& E_{2}^{(\alpha)}(\psi)\geq E_{2}^{(\alpha)}(\phi)\;\;\text{for  }\alpha>2
\end{align}
with equalities if and only if $\phi=\psi$ up to local unitaries.
\end{corollary}
The corollary follows directly from the proposition above and from theorem~\ref{lm}.

Note that the corollary above also implies that for $0<\alpha<2$, 
$$
E_{2}^{(\alpha)}(\psi)\leq E_{2}^{(\alpha)}(|M\rangle)\;,
$$
for all $\psi\in\mathcal{H}_4$ with $\tau_{(AB)(CD)}(\psi)$, $\tau_{(AC)(BD)}(\psi)$, and $\tau_{(AD)(BC)}(\psi)$
all being no greater than $4/3$. 

We end this section by classifying all the states in $\mathcal{M}$ for which one of the tangles $\tau_{(AB)(CD)}$, $\tau_{(AC)(BD)}$, and $\tau_{(AD)(BC)}$ is equal to $3/2$. Without loss of generality we will assume that $\tau_{(AB)(CD)}=3/2$.
\begin{proposition}
Let $\psi\in\mathcal{M}$. If $\tau_{(AB)(CD)}(\psi)=3/2$ then, up to local unitaries,
\begin{equation}\label{last}
|\psi\rangle=\frac{1}{2}u_0+\frac{i}{2}u_1+\frac{e^{i\theta}}{2}u_2+\frac{ie^{i\theta}}{2}u_3\;.
\end{equation}
\end{proposition}
\begin{remark}
The two cluster states with $\tau_{(AB)(CD)}=3/2$  have this form with $\theta=0$ or $\theta=\pi/2$.
Among all the states of this form, the choice $\theta=\pi/4$ gives the highest value for the average entropy
of entanglement, but it does not reach the average entropy of entanglement of the state $|M\rangle$ (see Fig.~2).

\begin{figure}[tp]
\includegraphics[scale=.48]{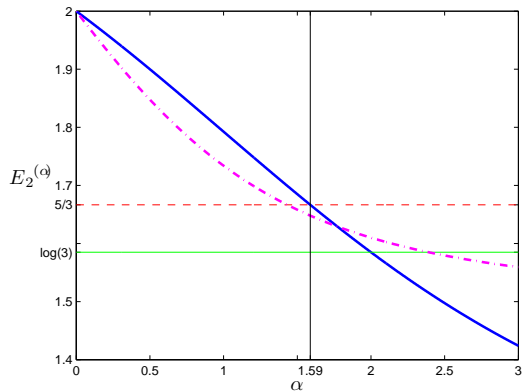}
\caption{A graph of the average Renyi $\alpha$-entropy of entanglement as a function of $\alpha$. The blue line corresponds to the state $|M\rangle$, the green line to the state $|L\rangle$, the dashed red line to the cluster states, and the dashed purple line to the state given in Eq.(\ref{last}) with $\theta=\pi/4$. For $\alpha\geq 2$ we proved that the cluster states maximize $E_{2}^{(\alpha)}$. However, we conjecture that the cluster states maximize $E_{2}^{(\alpha)}$ for $\alpha>\alpha_{0}\equiv 1.59...$ while the state $|M\rangle$ maximize it for $\alpha<\alpha_0$. Among all the states in $\mathcal{M}$ the state $|L\rangle$ and $|M\rangle$ minimize $E_{2}^{(\alpha)}$ for $0<\alpha<2$ and $\alpha>2$, respectively.}
\end{figure} 

\end{remark}
\begin{proof}
From  $\tau_{(AB)(CD)}=3/2$ we have
$$
\psi=\frac{1}{2}\sum_{j=0}^{3}e^{i\theta_j}u_j\;.
$$
Now from the condition in Eq.~(\ref{classm}), we get $\sum_{j=0}^{3}e^{i2\theta_j}=0$. This condition with the freedom
of global phase leads to the form in Eq.(\ref{last}).
\end{proof}

\section{Conclusions}\label{conclusions}
Four qubits entanglement is far more complicated to analyse than its three qubits counterpart.
This intricacy manifests itself with the uncountable number of inequivalent SLOCC classes.
Such complexity also occur in 5 and 6 qubits systems, although for these systems there exist maximally
entangled states (such as the 5-qubits code state) with the property that any bipartite cut yields 
a maximally entangled (bipartite) state. Since such states do not exists in 4-qubits nor in $n$-qubits
with $n\geq 8$, the study of 4-qubits entanglement gives an insight to the structure of $n$-qubits 
maximally entangled states with large $n$. Indeed, some of the results presented here, such as theorem~\ref{kn}, 
can be extended to $n$-qubits~\cite{WG}.

In this paper we found an operational interpretation for the 4-tangle as a kind of 4-party residual entanglement that
can not be shared between 2-qubits/2-qubits bipartite cuts. This operational interpretation enabled us to find
a family of maximally entangled states that is characterized by four real parameters. All the states in the family
maximize the average bipartite tangle, but only two states in the family (i.e. the states $|M\rangle$ and $|L\rangle$ in 
Eqs.(\ref{entropystates},\ref{statel})), 
maximize all the average Tsillas $\alpha$-entropy of entanglement. In this sense, up to local unitary transformations, 
there are only two maximally entangled four qubits states.

Both the states $|M\rangle$ and $|L\rangle$ are symmetric; that is, up to local unitaries, they 
are both invariant under permutations of the 4 qubits. The eigenvalues of their reduced density matrices that obtained
after tracing out two qubits have the form given in Theorem~\ref{lm}. Therefore, since they both maximize the average
bipartite tangle, we believe they also optimize many other averages of bipartite entanglement monotones that were 
not introduced here. Moreover, the techniques introduced here suggest that states with the properties of 
$|M\rangle$ and $|L\rangle$ may exists in higher dimensional systems~\cite{WG}.

We also found that the three cluster states in Eqs.(\ref{C1},\ref{C2},\ref{C3}) are the only (up to local unitaries) 4-qubits states 
that maximize the average tangle, and have the property that out of the three reduced density matrices, that obtained by tracing out  two qubits, two are proportional to the identity. In addition, we showed that the cluster states optimize the average 
Renyi $\alpha$-entropy
of entanglement with $\alpha\geq 2$. The reason that it is the cluster states and not $|M\rangle$ or $|L\rangle$
that optimize this average Renyi entropy, is that the Renyi entropy with $\alpha\geq 2$ is not concave and so the
Renyie entropy of entanglement is only a deterministic entanglement monotone and \emph{not} an ensemble monotone.

\emph{Acknowledgments:---}
The authors are grateful for Ben Fortescue for help with the graphs.
GG would like to thank Dominic Barry, Francesco Buscemi, Ben Fortescue, and Jeong San Kim for fruitful discussions.
GG research is supported by NSERC.

\begin{appendix}

\section{Kempf-Ness Theorem}\label{KempfNess}
The purpose of this appendix is to state the version of the Kempf-Ness theorem
that is used in this paper. Let $\mathcal{H}_{n}$ be $n$-qubit space and let
$G$ be the subgroup $SL(2,\mathbb{C})\otimes SL(2,\mathbb{C})\otimes\cdots\otimes SL(2,\mathbb{C})$
 ($n$-copies) in $GL(\mathcal{H}_{n})$. Let $K=$ $SU(2)\otimes
SU(2)\otimes\cdots\otimes SU(2)$. Let $\mathfrak{g}$ be the Lie algebra of $G$
contained in $End(\mathcal{H}_{n})$. We set $Crit(\mathcal{H}_{n})=\{\phi
\in\mathcal{H}_{n}|\left\langle \phi|X|\phi\right\rangle =0,X\in
\mathfrak{g\}}$. The Kempf-Ness theorem in this context says (the only hard
part of the theorem is the \textquotedblleft if\textquotedblright\ part of 3.
which we don't use in this paper):

\begin{theorem}
Let $\phi\in\mathcal{H}_{n}$ then

1. $\phi\in Crit(\mathcal{H}_{n}),g\in G$ then $\left\Vert g\phi\right\Vert
\geq\left\Vert \phi\right\Vert $ with equality if and only if $g\phi\in K\phi$.

2. If $\phi\in\mathcal{H}_{n}$ then $\phi\in Crit(\mathcal{H}_{n})$ if and
only if $\left\Vert g\phi\right\Vert \geq\left\Vert \phi\right\Vert $ for all
$g\in G$.

3. If $\phi\in\mathcal{H}_{n}$ then $G\phi$ is closed in $\mathcal{H}_{n}$ if
and only if $G\phi\cap Crit(\mathcal{H}_{n})\neq\emptyset$.
\end{theorem}

We now assume $n=4$. Let $\mathcal{A}$ be as in section II. Then

\begin{proposition}
$Crit(\mathcal{H}_{4})=K\mathcal{A}$.
\end{proposition}

\section{Proof of theorem~\ref{lm}}\label{aaa123}

We first prove the theorem for the von-Neumann entropy (i.e. the case $\alpha=1$) and then
we will consider the case $\alpha\neq 1$ separately.
We want to optimize the function
$
f(\lambda_0,\lambda_1,\lambda_2,\lambda_3)=-\sum_{k=0}^{3}\lambda_k\log\lambda_k
$
under the constraints $\sum_{k=0}^{3}\lambda_k=1$ and $\sum_{k=0}^{3}\lambda_{k}^{2}=1-\tau/2$, while
$0\leq \lambda_k\leq 1$. 
The Lagrangian is therefore given by
$$
\mathcal{L}=-\sum_{k=0}^{3}\lambda_k\log\lambda_k+\mu\left(\sum_{k=0}^{3}\lambda_k-1\right)+
\nu\left(\sum_{k=0}^{3}\lambda_{k}^{2}+\frac{\tau}{2}-1\right)\;,
$$
where $\mu$ and $\nu$ are the Lagrange multipliers. Therefore, the critical points in the interior of the domain 
(i.e. $0<\lambda_k<1$) must satisfies the equation:
\begin{equation}\label{a1a}
\frac{\partial\mathcal{L}}{\partial\lambda_k}=-\log\lambda_k-\log e+\mu+2\nu\lambda_k=0
\end{equation}
Now, we first show that if all for $\lambda_k$ satisfy the equation above, then the set $\{\lambda_k\}$ contains
at most two distinct numbers. To see that, suppose that there are three distinct numbers. Then, without loss of generality,
lets assume that $\lambda_0>\lambda_1>\lambda_2>0$. Thus, from the 3 equations above (for $k=0,1,2$) it follows that
$$
(\lambda_0-\lambda_1)\log\lambda_2+(\lambda_1-\lambda_2)\log\lambda_0=(\lambda_0-\lambda_2)\log\lambda_1
$$
Denote by $a\equiv (\lambda_1-\lambda_2)/(\lambda_0-\lambda_1)$. Hence, $a>0$ and
$$
\log\lambda_2+a\log\lambda_0=(1+a)\log\lambda_1\;,
$$
which is equivalent to
$$
\lambda_2\lambda_{0}^{a}=\lambda_{1}^{1+a}.
$$
Denote by $x\equiv\lambda_2/\lambda_1$ and $y\equiv \lambda_0/\lambda_1$. Therefore, $x<1$, $y>1$,  
$a=(1-x)/(y-1)$, and $$xy^a=1.$$ From the last equation and the generalized arithmetic-geometric mean inequality
we get
$$
1=(xy^a)^{1/(1+a)}\leq\frac{1}{1+a}\left(x+ay\right)=1\;,
$$
where the last equality is obtained by substituting $a=(1-x)/(y-1)$.
Not that the geometric-arithmetic mean inequality is saturated if and only if $x=y$ and therefore
we get a contradiction to the assumption that there are 3 distinct numbers in the set $\{\lambda_i\}$.
Therefore, for the interior points we have 3 options: \textbf{(a)} $\lambda_0\geq\lambda_1=\lambda_2=\lambda_3$,
\textbf{(b)} $\lambda_0=\lambda_1\geq\lambda_2=\lambda_3$ and 
\textbf{(c)} $\lambda_0=\lambda_1=\lambda_2\geq\lambda_3$. Option $\textbf{(b)}$ is the only one that does
not appear in theorem~\ref{lm}. In this case $\tau$ must be greater than 1, and 
$\lambda_0=\lambda_1=\frac{1+\sqrt{3}x(\tau)}{4}$ and $\lambda_2=\lambda_3=\frac{1-\sqrt{3}x(\tau)}{4}$. 
It is a straightforward calculation to show that the von-Neumann entropy of this distribution never equals the von-Neumann entropy of the distributions that appear in the theorem. (Note that this is all we have to show since it has 
already been proved in~\cite{Ber03} that the distributions in the theorem are the optimal ones).

 As for the critical points on the boundary, set $\lambda_3=0$ and then the same argument as above implies that
 the set $\{\lambda_0,\lambda_1,\lambda_2\}$ contains at most two distinct numbers. We therefore have two options:
 \textbf{(a)} $\lambda_0=\lambda_1\geq\lambda_2$ and
 \textbf{(b)} $\lambda_0\geq\lambda_1=\lambda_2$. Again, the distribution (b) does not appear in the theorem,
 but it is straightforward to show that its von-Neumann entropy never equals to the von-Neumann entropies of the
 distributions in the theorem. The last point on the bounday that we need to check is when $\lambda_3=\lambda_2=0$,
 but this distribution appears in the theorem. This completes the proof for the case $\alpha=1$.
 
 We now prove the theorem for the case $\alpha\neq 1$ (as well as $\alpha\neq 2$). 
 In this case, we optimize the function
$
f(\lambda_0,\lambda_1,\lambda_2,\lambda_3)=\sum_{k=0}^{3}\lambda_{k}^{\alpha}
$
under the same constraints above; that is, $\sum_{k=0}^{3}\lambda_k=1$ and $\sum_{k=0}^{3}\lambda_{k}^{2}=1-\tau/2$. 
The Lagrangian in this case is given by
$$
\mathcal{L}=\sum_{k=0}^{3}\lambda_{k}^{\alpha}+\mu\left(\sum_{k=0}^{3}\lambda_k-1\right)+
\nu\left(\sum_{k=0}^{3}\lambda_{k}^{2}+\frac{\tau}{2}-1\right)\;,
$$
where $\mu$ and $\nu$ are the Lagrange multipliers. The critical points in the interior of the domain 
must satisfies the equation:
\begin{equation}\label{a2a}
\frac{\partial\mathcal{L}}{\partial\lambda_k}=\alpha\lambda_{k}^{\alpha-1}-\log e+\mu+2\nu\lambda_k=0
\end{equation} 
Similarly to the argument above, we first show that the set $\{\lambda_k\}$ contains
at most two distinct numbers. To see that, suppose that there are three distinct numbers
$\lambda_0>\lambda_1>\lambda_2>0$. Thus, from the 3 equations above (for $k=0,1,2$) it follows that
$$
\frac{\lambda_{0}^{\alpha-1}-\lambda_{1}^{\alpha-1}}{\lambda_{0}^{\alpha-1}-\lambda_{2}^{\alpha-1}}=
\frac{\lambda_0-\lambda_1}{\lambda_0-\lambda_2}\;.
$$
Now, denote by $x\equiv\lambda_1/\lambda_0$ and $y\equiv\lambda_2/\lambda_0$. From our assumptions
$0<y<x<1$. With these notations the equation above can be written as:
$$
\frac{1-x^{\alpha-1}}{1-x}=\frac{1-y^{\alpha-1}}{1-y}\;.
$$
However, for non-negative $\alpha\neq 2$ the function $f(x)=(1-x^{\alpha-1})/(1-x)$ is one-to-one and therefore
we get a contradiction. This complete the proof that the set $\{\lambda_k\}$ contains at most two distinct numbers.
The rest of the proof follows the same lines as the proof for the case $\alpha=1$.
 
 \section{Global max of $f_{\max}^{(\alpha)}$ and global min of $f_{\min}^{(\alpha)}$}\label{optimization}
 
\subsection{Calculation of $U^{(\alpha)}$}
In this section we prove that for $\alpha>2$, $U^{(\alpha)}=f_{\max}^{(\alpha)}(4/3,4/3,4/3)$, and $(4/3,4/3,4/3)$ is the only point of global maximum.
Since the function $f_{\max}^{(\alpha)}(t_1,t_2,t_3)$ is invariant under permutations of $t_1$, $t_2$, and $t_3$, 
we can assume without loss of generality that the global maximum
of $f_{\max}^{(\alpha)}$ is obtained at a point with $t_1\geq t_2\geq t_3$.
We will therefore look for the maximum of $f_{\max}^{(\alpha)}(t_1,t_2,4-t_1-t_2)$ in the domain
$$
D=\left\{(t_1,t_2)\Big|\frac{4}{3}\leq t_1\leq\frac{3}{2}\;,\;2-\frac{1}{2}t_1\leq t_2\leq t_1\right\}.  
$$
In this domain, $t_2$ can be either bigger or smaller than $4/3$ and therefore the derivative
of $f_{\max}^{(\alpha)}$ with respect to $t_2$ is not continuous at points with $t_2=4/3$. 
We therefore split the domain
$D$ into two regions 
\begin{align}
& D_1=\left\{(t_1,t_2)\Big|\frac{4}{3}\leq t_1\leq\frac{3}{2}\;,\;\frac{4}{3}\leq t_2\leq t_1\right\}\nonumber\\ 
& D_2=\left\{(t_1,t_2)\Big|\frac{4}{3}\leq t_1\leq\frac{3}{2}\;,\;2-\frac{1}{2}t_1\leq t_2\leq \frac{4}{3}\right\}\nonumber  
\end{align}
so that on $D_1$ (or $D_2$) all the derivatives of $f_{\max}^{(\alpha)}$ are continuous.
We start by maximizing $f_{\max}^{(\alpha)}$ on the domain $D_1$. 

\subsubsection{Maximizing $f_{\max}^{(\alpha)}$ on the domain $D_1$}

Denote by $x_i=\sqrt{1-\frac{2}{3}t_i}$ for $i=1,2$ and $y=\sqrt{1-\frac{3}{4}t_3}$.
In these variables, the function $g_{\max}^{(\alpha)}(x_1,x_2)\equiv f_{\max}^{(\alpha)}(t_1,t_2,4-t_1-t_2)$
is given by
\begin{align}
g_{\max}^{(\alpha)}(x_1,x_2) =\frac{1}{\alpha-1} & \Big[1-\sum_{i=1,2}\left(\frac{(1+x_i)^\alpha}{4^\alpha}+\frac{1}{3}
\frac{\left(1-3x_i\right)^\alpha}{4^\alpha}\right)\nonumber\\
& -\frac{2}{3^{\alpha+1}}(1+y)^\alpha-\frac{1}{3^{\alpha+1}}(1-2y)^\alpha
\Big],\nonumber
\end{align}
where in term of the variables $x_1$ and $x_2$, $y=\frac{1}{2}\sqrt{1-\frac{9}{2}(x_1^2+x_2^2)}$.
In terms of these new variables, the domain $D_1$ is given by $0\leq x_1\leq x_2\leq 1/3$. We start by looking at the critical points in the interior of $D_1$. 

The critical points of $g_{\max}^{(\alpha)}(x_1,x_2)$ satisfies the conditions
\begin{align}
& \frac{\partial g_{\max}^{(\alpha)}}{\partial x_1}(x_1,x_2)=\frac{x_1}{4}\left(u_\alpha(y)-v_\alpha(x_1)\right)=0\nonumber\\
& \frac{\partial g_{\max}^{(\alpha)}}{\partial x_2}(x_1,x_2)=\frac{x_2}{4}\left(u_\alpha(y)-v_\alpha(x_2)\right)=0\;,\nonumber
\end{align}
where
\begin{align}
& v_\alpha (x)\equiv\frac{\alpha}{(\alpha-1)4^{\alpha-1}}\frac{1}{x}\left[(1+x)^{\alpha-1}-(1-3x)^{\alpha-1}\right]\nonumber\\
& u_\alpha (y)\equiv\frac{\alpha}{(\alpha-1)3^{\alpha-1}}\frac{1}{y}\left[(1+y)^{\alpha-1}-(1-2y)^{\alpha-1}\right]\;.\nonumber
\end{align}
Hence, the point $(x_1,x_2)$ is critical if and only if $v_\alpha(x_1)=v_\alpha(x_2)=u_\alpha(y)$, where 
$y=\frac{1}{2}\sqrt{1-\frac{9}{2}(x_1^2+x_2^2)}$ (note that $0\leq y\leq 1/2$). In the next two lemmas we prove two useful properties of the 
functions $v_\alpha(x)$ and $u_\alpha(y)$.
\begin{lemma}\label{ab}
If $2<\alpha< 4$, and $0\leq y\leq 1/2$, then $u_{\alpha}'(y)<0$. If $2<\alpha< 5$, and $0\leq x\leq 1/3$, then 
$v_{\alpha}'(x)<0$.
\end{lemma}
\begin{proof}
A simple calculation gives
\begin{align}
& u_{\alpha}'(y) =\nonumber\\
& \frac{\alpha}{(\alpha-1)3^{\alpha-1}}\frac{1}{y^2}
\left[(1+y)^\beta(\beta y -1)+(1-2y)^\beta (1+2\beta y)\right],\nonumber
\end{align}
where $\beta\equiv\alpha-2$. From the assumption of the lemma $0<\beta<2$. Therefore, one can easily check that
$u_\alpha '(0)<0$ and $u_\alpha '(1/2)<0$. All that is left to show is that in the domain $0<y<1/2$ the function
$$w(y)\equiv(1+y)^\beta(\beta y -1)+(1-2y)^\beta (1+2\beta y)$$ is negative.  To find its maximum value, we calculate
its critical points. The requirement $w'(y)=0$ gives $(1+y)^{\beta-1}=4(1-2y)^{\beta-1}$. Clearly, there are no critical points for $\beta\leq 1$ in the domain $(0,1/2)$. For $\beta>1$ we express the value of $w(y_c)$ at the critical point 
by substituting for $(1+y_c)^{\beta-1}$ the value $4(1-2y_c)^{\beta-1}$. This gives,
$$
w(y_c)=-3(1-2y_c)^{\beta-1}\left[1-2(\beta-1)y_c\right]<0\;,
$$
for $y_c<1/2$ and $\beta<2$.
Hence, $u_{\alpha}'(y)<0$ for $2<\alpha<4$. Following the same arguments, one can show that $v_{\alpha}'(x)<0$
for $2<\alpha<5$.
\end{proof}
\begin{lemma}\label{bb}
If $\alpha\geq 4$ then the global maximum of  $v_\alpha(x)$ (in the domain $0\leq x\leq 1/3$) is strictly smaller than
the global minimum of the function $u_\alpha(y)$ (in the domain $0\leq y\leq 1/2$).
\end{lemma}
\begin{proof}
The global extremum points of $u_\alpha$ and $v_\alpha$ are obtained on the boundary or on critical points.
Therefore, we first check the bounday (that is, end points). We have 
\begin{align}
& v_{\alpha}(0)=\frac{\alpha}{4^{\alpha-2}}\;\;,\;\;v_{\alpha}\left(\frac{1}{3}\right)=\frac{\alpha}{\alpha-1}\frac{1}{3^{\alpha-2}}\nonumber\\
& u_{\alpha}(0)=\frac{\alpha}{3^{\alpha-2}}\;\;,\;\;u_{\alpha}\left(\frac{1}{2}\right)=\frac{\alpha}{\alpha-1}\frac{1}{2^{\alpha-2}}\nonumber
\end{align}
Clearly, for $\alpha\geq 4$, we have $\max\{v_{\alpha}(0),v_{\alpha}(1/3)\}<\min\{u_{\alpha}(0),u_{\alpha}(1/2)\}$.
We now estimate the values of $u_\alpha$ and $v_\alpha$ at their critical points $x_c$ and $y_c$. 
From $u_{\alpha}'(y_c)=0$ and $v_{\alpha}'(x_c)=0$ we have
\begin{align}
& u_{\alpha}(y_c)=\frac{\alpha}{3^{\alpha-1}}\left[(1+y_c)^{\alpha-2}+2(1-2y_c)^{\alpha-2}\right]\nonumber\\
& v_{\alpha}(x_c)=\frac{\alpha}{4^{\alpha-1}}\left[(1+x_c)^{\alpha-2}+3(1-3x_c)^{\alpha-2}\right]\;.
\label{yc}
\end{align}
Since we do not have explicit expressions for $x_c$ and $y_c$, we find an upper bound for
$v_{\alpha}(x_c)$ and a lower bound for $u_{\alpha}(y_c)$. 
Since the functions in Eq.~(\ref{yc}) are convex for $\alpha\geq 4$, we get
\begin{align}
v_{\alpha}(x_c) & \leq\max\{v_{\alpha}(x_c=0),\;v_{\alpha}(x_c=1/3)\}\nonumber\\
& =\max\Big\{\frac{1}{4}\frac{\alpha}{3^{\alpha-2}}\;,\;\frac{\alpha}{4^{\alpha-2}}\Big\}\nonumber
\end{align}
The minimum value of the function $u_{\alpha}(y_c)$ given in Eq.~(\ref{yc}) is obtained at the point
$$
y_c=\frac{4^{1/(\alpha-3)}-1}{1+2\cdot 4^{1/(\alpha-3)}}\;.
$$
Note that this is not necessarily the true value of $y_c$, but rather the value at which the function in Eq.(\ref{yc})
is minimized. It is a straightforward calculation to show that at this value of $y_c$
$$
u_{\alpha}(y_c)>\max\Big\{\frac{1}{4}\frac{\alpha}{3^{\alpha-2}}\;,\;v_\alpha(0)\;,\;v_\alpha(1/3)\Big\}\;.
$$
This completes the proof that $u_\alpha(y)>v_\alpha(x)$ for $\alpha\geq 4$ and for all $x\in[0,1/3]$ and $y\in[0,1/2]$.
\end{proof}

From lemma~\ref{bb} it follows that for $\alpha\geq 4$ the function $g_{\max}^{(\alpha)}(x_1,x_2)$ does not have critical points in the interior of $D_1$. For $2<\alpha<4$, $g_{\max}^{(\alpha)}(x_1,x_2)$ can have critical points. However,
from the second derivatives test and from lemma~\ref{ab}, it follows that the Hessian is positive definite. Therefore,
these critical points are local min and can not be a global max. The global maximum of $g_{\max}^{(\alpha)}(x_1,x_2)$
is therefore obtained at the boundary of $D_1$.

The boundary of $D_1$ is a triangle with 3 sides given by $x_1=0$, $x_1=x_2$, and $x_2=1/3$.
If $x_1=0$ then
$$
\frac{d g_{\max}^{(\alpha)}(0,x_2)}{d x_2}=\frac{x_2}{4}\left(u_\alpha(y)-v_\alpha(x_2)\right) \;.
$$
Therefore, $g_{\max}^{(\alpha)}(0,x_2)$ is convex for $2<\alpha<4$ (see lemma~\ref{ab}), 
and has no critical points for $\alpha\geq 4$ (see lemma~\ref{bb}). Therefore, its global max is obtained
at one of the end points $(0,0)$ or $(0,1/3)$. 

On the side $x_1=x_2\equiv x$ we have
$$
\frac{d}{d x}g_{\max}^{(\alpha)}(x,x)=\frac{x}{2}\left(u_\alpha(y)-v_\alpha(x)\right) \;.
$$
Hence, the same arguments implies that the global max of $g_{\max}^{(\alpha)}(x,x)$ is obtained
at one of the end points $x=0$ or $x=1/3$. Similarly, on the side $x_2=1/3$, the function 
$g_{\max}^{(\alpha)}(x_1,1/3)$ obtains its global maximum at one of the end points $x_1=0$ or $x_1=1/3$.  
Among the three vertices $(0,0),\;(0,1/3),\;(1/3,1,3)$,  we have
$$
g_{\max}^{(\alpha)}(1/3,1/3)> \max\big\{g_{\max}^{(\alpha)}(0,0)\;,\;g_{\max}^{(\alpha)}(0,1/3)\big\}
$$
for all $\alpha>2$.
Hence, on $D_1$, $g_{\max}^{(\alpha)}$ obtains its global max at the point $(x_1,x_2)=(1/3,1/3)$ which is equivalent
to $t_1=t_2=t_3=4/3$.

\subsubsection{Maximizing $f_{\max}^{(\alpha)}$ on the domain $D_2$}
Denote by $x=\sqrt{1-\frac{2}{3}t_1}$ and $y_i=\sqrt{1-\frac{3}{4}t_i}$ for $i=2,3$ .
In these variables, the function $h_{\max}^{(\alpha)}(y_2,y_3)\equiv f_{\max}^{(\alpha)}(4-t_2-t_3,t_2,t_3)$
is given by
\begin{align}
h_{\max}^{(\alpha)}(y_2,y_3)=&\nonumber\\
\frac{1}{\alpha-1}\Big[ & 1-\frac{2}{3^{\alpha+1}}\sum_{i=2,3}\left((1+y_i)^\alpha-\frac{1}{2}(1-2y_i)^\alpha\right)\nonumber\\
& \;\;\;\;\;\;\;\;\;\;\;\;-\frac{1}{4^\alpha}\left((1+x)^\alpha+\frac{1}{3}
\left(1-3x\right)^\alpha\right)
\Big],\nonumber
\end{align}
where in term of the variables $y_2$ and $y_3$, $x=\frac{1}{3}\sqrt{1-8(y_2^2+y_3^2)}$.
In terms of these new variables, the domain $D_2$ is therefore given by 
$$
D_2=\Big\{(y_2,y_3)\Big|y_2^2+y_3^2\leq\frac{1}{8}\;\;,\;\;y_3\geq y_2\geq 0\Big\}\;.
$$ 
We first look at the critical points in the interior of $D_2$. 

The critical points of $h_{\max}^{(\alpha)}(y_2,y_3)$ satisfies the conditions
\begin{align}
& \frac{\partial h_{\max}^{(\alpha)}}{\partial y_2}(y_2,y_3)=\frac{2y_2}{9}\left(v_\alpha(x)-u_\alpha(y_2)\right)=0\nonumber\\
& \frac{\partial h_{\max}^{(\alpha)}}{\partial y_3}(y_2,y_3)=\frac{2y_3}{9}\left(v_\alpha(x)-u_\alpha(y_3)\right)=0\;.\nonumber
\end{align}
Hence, it follows from lemma~\ref{ab} and lemma~\ref{bb} that the global max of $h_{\max}^{(\alpha)}(y_2,y_3)$
is obtained on the boundary of $D_2$.

The boundary of $D_2$ consists of the line $y_2=0$, the line $y_1=y_2$, and the curve 
$(y_2,y_3)=(\sin\theta/\sqrt{8},\cos\theta/\sqrt{8})$ with $0\leq\theta\leq\pi/4$. The global maximum of
$h_{\max}^{(\alpha)}$ on the lines $y_2=0$ and $y_1=y_2$ is obtained on one of the endpoints
$(0,1/2)$, $(0,0)$, and $(1/4,1/4)$. The argument follows from lammas~\ref{ab} and~\ref{bb}, 
in the same way as it was used in the analysis of the boundary of $D_1$. Therefore, we focus now
on the curve $(y_2,y_3)=(\sin\theta/\sqrt{8},\cos\theta/\sqrt{8})$ with $0\leq\theta\leq\pi/4$.

Let
$$
h(\theta)\equiv h_{\max}^{(\alpha)}\left(\frac{\sin\theta}{\sqrt{8}},\frac{\cos\theta}{\sqrt{8}}\right).
$$
Note that for these values of $y_2$ and $y_3$, $x=0$. Hence, 
$$
h'(\theta)=\frac{\sin(2\theta)}{72}\left[u_\alpha\left(\frac{\cos\theta}{\sqrt{8}}\right)-
u_\alpha\left(\frac{\sin\theta}{\sqrt{8}}\right)\right]
$$
From lemma~\ref{ab} the function $u_\alpha$ is one-to-one for $2<\alpha<4$. Therefore, the only
critical point in this case is $(y_2,y_3)=(1/4,1/4)$. For $\alpha\geq 4$ it is a simple calculation to verify
that $h(\theta)<h_{\max}^{(\alpha)}(0,0)$. Therefore, since
$$
h_{\max}^{(\alpha)}(0,0)> \max\big\{h_{\max}^{(\alpha)}(1/4,1/4)\;,\;h_{\max}^{(\alpha)}(0,1/2)\big\}
$$
for all $\alpha>2$,
we conclude that on $D_2$, $h_{\max}^{(\alpha)}$ obtains its global max at the point $(y_1,y_2)=(0,0)$ which is equivalent to $t_1=t_2=t_3=4/3$.

\subsection{Calculation of $L^{(\alpha)}$}
In this section we prove that for $0<\alpha<2$, $L^{(\alpha)}=f_{\min}^{(\alpha)}(4/3,4/3,4/3)$, and 
$(4/3,4/3,4/3)$ is the only point of global minimum. From theorem~\ref{lm} it follows that 
$f_{\min}^{(\alpha)}(t_1,t_2,t_3)$ for $0<\alpha<2$ is given
by the exact same expression as $f_{\max}^{(\alpha)}(t_1,t_2,t_3)$ for $\alpha>2$.
Therefore, our proof that $L^{(\alpha)}=f_{\min}^{(\alpha)}(4/3,4/3,4/3)$ follows the exact same steps
used in the calculation of $U^{(\alpha)}$ for $\alpha>2$. The only difference is that the lemmas~\ref{ab} and~\ref{bb}
do not hold for $0<\alpha<2$, and instead we have the following lemma.
\begin{lemma}
If $0<\alpha<2$, $0\leq x\leq1/3$, and $0\leq y\leq 1/2$, then $u_\alpha '(y)>0$ and $v_\alpha '(x)>0$.
\end{lemma}
\begin{proof}
A simple calculation gives
$$
u_{\alpha}'(y)=\frac{\alpha}{(\alpha-1)3^{\alpha-1}}\frac{1}{y^2}\left[
\frac{(1+2\beta y)}{(1-2y)^\beta}-\frac{(1+\beta y)}{(1+y)^\beta}\right]\;,
$$
where $\beta\equiv 2-\alpha$. From the assumption of the lemma $0<\beta<2$. Therefore, one can easily check that
$\lim_{y\rightarrow 0}u_\alpha '(y)>0$ and $\lim_{y\rightarrow 1/2}u_\alpha '(y)=+\infty$. 
All that is left to show is that in the domain $0<y<1/2$ the function
$$
w(y)\equiv\frac{1}{(1-\beta)}\frac{1}{y^2}\left[
\frac{(1+2\beta y)}{(1-2y)^\beta}-\frac{(1+\beta y)}{(1+y)^\beta}\right]
$$ 
is positive.  To find its minimum value, we would like to calculate
its critical points. However, the requirement $w'(y)=0$ gives $4(1+y)^{\beta+1}=(1-2y)^{\beta+1}$. Hence, there are no critical points in the domain $(0,1/2)$. That is,, $u_{\alpha}'(y)>0$ for $0<\alpha<2$. Following the same arguments, one can show that $v_{\alpha}'(x)>0$
for $0<\alpha<2$.
\end{proof}

With this lemma replacing lemmas~\ref{ab} and~\ref{bb}, the proof that $L^{(\alpha)}=f_{\min}^{(\alpha)}(4/3,4/3,4/3)$
follows exactly the same steps that appear in the calculation of $U^{(\alpha)}$.
 
\end{appendix}

\end{document}